\documentclass{eptcs}
\providecommand{\event}{G.Ciobanu, M.Koutny (Eds.): Membrane Computing and\\
Biologically Inspired Process Calculi (MeCBIC 2010) 
EPTCS Proceedings, 2010.} 

\usepackage{epic}
\usepackage{eepic}
\usepackage{latexsym}
\usepackage{graphicx}
\usepackage{amsmath}
\usepackage{kappa}
\usepackage{crosstalk}

\newcommand{\interface}{\Sigma}
\newcommand{\stateset}{\mathbb{I}}
\newcommand{\bindinglabel}{\mathbb{N}}
\newcommand{\strutfrac}[2]{\frac{\strut #1}{\strut #2}}
\newcommand{\rlhs}{lhs}

\newcommand{\sat}{\models}
\newcommand{\embed}[3]{#1 \lhd_{#3} #2} 

\newcommand{\auto}[1]{\embedset{#1}{#1}}
\newcommand{\nauto}[1]{|\auto{#1}|}
\newcommand{\logicand}{\wedge}
\newcommand{\embedset}[2]{[#1,#2]}

\def\gr#1{\hat r} 
\def\emph#1{\textit{#1}}
\def\rar{\rightarrow}
\def\Rar{\Rightarrow}

\def\ga{\gamma}

\def\lrar{\leftrightarrow}

\def\ga{\gamma}

\def\xAbs{\tilde{x}}
\def\lAbs{\tilde{l}}
\def\wAbs{\tilde{w}}

\def\WAbs{\tilde{\cal W}}

\def\PI{\cal P}

\def\L{{\mathcal L}}
\def\S{{\mathcal X}}

\def\Tauu{\tau_{\IReal}}
\def\tildtauu{\tilde{\tau}_{\IReal}}
\def\tildetset{\mathcal{T}_{\IReal}(\tilde{\mathcal{W}})}
\def\tset{\mathcal{T}_{\IReal}(\mathcal{W})}

\def\R{{\mathbb R}}

\def\ra{\rightarrow}

\newcommand{\ptraceTanja}[1][k]{\pi_0(x_0) \prod_{i=1}^k \frac{w(x_{i-1},l_i,x_i)}{a(x_{i-1})}
\left (e^{-a(x_{i-1})\cdot \inf(I_i)}-e^{-a(x_{i-1})\cdot \sup(I_i)}\right)}

\def\longrightharpoonup{\relbar\joinrel\rightharpoonup}
\def\longleftharpoondown{\leftharpoondown\joinrel\relbar}

\def\longrightleftharpoons{
  \mathop{
    \vcenter{
      \hbox{
    \ooalign{
      \raise1pt\hbox{$\longrightharpoonup\joinrel$}\crcr
      \lower1pt\hbox{$\longleftharpoondown\joinrel$}
    }
      }
    }
  }
}

\newcommand{\emptystate}{\epsilon}
\newcommand{\emptybinding}{\epsilon}
\newcommand{\emptylist}{\varepsilon}

\newcommand{\emptyagent}{\emptyset}

\newcommand{\agentname}{\mathcal {A}}
\newcommand{\sitename}{\mathcal {S}}

\newcommand{\transarrow}[1][]{%
\def\arga{#1}%
\def\argb{\sharp}%
\ifx\arga\argb
\rightsquigarrow
\else%
\rightarrow_{#1}
\fi}

\newcommand{\trans}[1][]{%
\def\argI{#1}%
\transrelay
}

\newcommand{\transrelay}[4][]{#2 \overset{#3%
{\def\arga{}
\ifx\arga\argI
\else%
{,\argI}%
\fi}%
}{\transarrow[#1]}#4}
\usepackage{mathrsfs}
\usepackage{amsmath, amsthm, amssymb}

\usepackage{NamedColors}
\usepackage[dvipsnames]{color}
\usepackage{subfigure}
\usepackage{graphicx}
\usepackage{picinpar}
\usepackage{floatflt}
\usepackage{epsfig}

\def\RB{\mathbb R}
\def\RR{\mathbb R}

\newcommand{\IReal}{{ \mathbb{IR} }}

\def\R{{\mathbf R}}
\def\Pr{{\mathrm Pr}}

\def\ra{\rightarrow}

\newenvironment{myitemize}{\begin{itemize}\setlength{\parsep}{0pt}\setlength{\itemsep}{0pt}}{\end{itemize}}
\newenvironment{myenumerate}{\begin{enumerate}\setlength{\parsep}{0pt}\setlength{\itemsep}{0pt}}{\end{enumerate}}
\newtheorem{definition}{Definition}
\newtheorem{thm}{Theorem}
\newtheorem{lemma}{Lemma}

\title{Lumpability Abstractions of Rule-based Systems\thanks{J{\'e}r{\^o}me Feret's contribution was partially supported by the \textsc{AbstractCell} ANR-Chair of Excellence. Heinz Koeppl acknowledges the support from the Swiss National Science Foundation, grant no. 200020-117975/1. Tatjana Petrov acknowledges the support from SystemsX.ch, the Swiss Initiative in Systems Biology.}}
\author{Jerome Feret
\institute{ LIENS (INRIA/{\'E}NS/CNRS) \\ Paris, France}
\email{feret@ens.fr}
\and
Thomas Henzinger 
\institute{Institute and Science of Technology\\
Vienna, Austria}
\email{thenzinger@ist.ac.at}
\and
Heinz Koeppl  
\institute{School of Computer and Communication Sciences\\EPFL\\
Lausanne, Switzerland}
\email{heinz.koeppl@epfl.ch}
\and
Tatjana Petrov 
\institute{School of Computer and Communication Sciences\\EPFL\\
Lausanne, Switzerland}
\email{tatjana.petrov@epfl.ch}
}
\def\titlerunning{Lumpability Abstractions of Rule-based Systems}
\def\authorrunning{J. Feret, T. Henzinger, H. Koeppl, T. Petrov}
\begin{document}
\maketitle

\begin{abstract}
The induction of a signaling pathway is characterized by transient complex formation and mutual posttranslational modification of proteins. 
To faithfully capture this combinatorial process in a mathematical model is an important challenge in systems biology. 
Exploiting the limited context on which most binding and modification events are conditioned, attempts have been made to reduce the combinatorial complexity by quotienting the reachable set of molecular species, into species aggregates while preserving the deterministic semantics of the thermodynamic limit. 
Recently we proposed a quotienting that also preserves the stochastic semantics and that is complete in the sense that the semantics of individual species can be recovered from the aggregate semantics. 
In this paper we prove that this quotienting yields a sufficient condition for \emph{weak lumpability} and that it gives rise to a \emph{backward Markov bisimulation} between the original and aggregated transition system.     
We illustrate the framework on a case study of the EGF/insulin receptor crosstalk.
\end{abstract}

\section{Introduction}
\label{intro}

Often a few elementary events of binding and covalent modification \cite{walshct_2006} in a biomolecular reaction system give rise to a combinatorial number of non-isomorphic reachable species or complexes \cite{hlavacekws_2003,hlavacek2006rms}. Instances of such systems are signaling pathway, polymerizations involved in cytoskeleton maintainance, the formation of transcription factor complexes in gene-regulation. 

For such biomolecular systems, traditional chemical kinetics face fundamental limitations, that are related to the question of how biomolecular events are represented and translated into a mathematical model \cite{mcquarrieda_1967_1}. More specifically, chemical reactions can only operate on a collection of fully specified molecular species and each such species gives rise to one differential equation, describing the rate of change of that species' concentration. Many combinatorial systems do not permit the enumeration of all molecular species and thus render their traditional differential description prohibitive. However, even if one could enumerate them, it remains questionable whether chemical reactions are the appropriate way to represent and to reason about such systems.

As the dynamics of a biomolecular reaction mixture comes about through the repeated execution of a few elementary events one may wonder about the effective degrees of freedom of the reaction mixture's dynamics. If the velocity of all events -- or their probabilities to occur per time-unit per instance -- are different for all complexes (w.r.t. modification) and pairs of complexes (w.r.t. binding) to which the events can apply to, then the degrees of freedom would equal to the number of molecular species. However, due to the local nature of physical forces underlying molecular dynamics, the kinetics of most events appear to be ignorant with respect to the global configuration of the complexes they are operating on. More provocatively, one may say that even if there would be variations of kinetics of an event from one context to another, experimental biology does not -- and most likely never will -- have the means to discern between all different contexts. For instance, fluorescence resonance energy transfer (FRET), may report on a specific protein-binding event and even its velocity, however we have no means to determine whether the binding partners are already part of a protein complex -- not to speak of the composition and modification state of these complexes. To this end, molecular species remain elusive and appear to be inappropriate entities of descriptions.

To align with the mentioned experimental insufficiencies and with the underlying biophysical locality, rule-based or agent-based descriptions were introduced as a framework to encode such reaction mixtures succinctly and to enable their mathematical analysis \cite{kappa,blinov2004}. Rules exploit the limited context on which most elementary events are conditioned. They just enumerate that part of a molecular species that is relevant for a rule to be applicable. Thus, in contrast to chemical reactions, rules can operate on a collection of partially specified molecular species. Recently, attempts have been made to identify the set of those partially specified species that allow for a self-consistent description of the rule-set's dynamics \cite{pnas,conzelmann2008}. Naturally, as partially specified species -- or \emph{fragments} -- in general encompass many fully specified species, the cardinality of that set is less than of the set of molecular species. These approaches aim to obtain a self-consistent fragment dynamics based on ordinary differential equations. It represents the dynamics in the thermodynamic limit of stochastic kinetics when scaling species multiplicities to infinity while maintaining a constant concentration (multiplicity per unit volume) \cite{kurtztg_1971_1}. In many applications in cellular biology this limiting dynamics is an inappropriate model due to the low multiplicities of some molecular species -- think of transcription factor - DNA binding events. We recently showed that the obtained \emph{differential fragments} cannot be used to describe the finite volume case of stochastic kinetics \cite{journal}. Exploiting statistical independence of events we instead derived \emph{stochastic fragments} that represent the effective degrees of freedom in the stochastic case. Conceptually, the procedure replaces the Cartesian product by a Cartesian sum for statistically independent states. In contrast to the differential case, stochastic fragments have the important property that the sample paths of molecular species can be recovered from that of partially specified species. 

We believe that interdisciplinary fields, such as systems biology, can move forward quickly enough only if the well-established knowledge in each of the disciplines involved is exploited to its maximum, i.e.~if the standardized, well-established theories are recognized and re-used.  For that reason, in this paper we translate our abstraction method (\cite{journal}) into the language of well-established contexts of abstraction for probabilistic systems -- lumpability and bisimulation. Lumpability is mostly considered from a theoretical point of view in the theory of stochastic processes \cite{lumpCommutCTMCExact,unif1,unif2,weakLumpDTMC, weakLumpCTMC,buchholz_lump}. A Markov chain is lumpable with respect to a given aggregation (quotienting) of its states, if the lumped chain preserves the Markov property \cite{KS60:lump}. A sound aggregation for \textit{any} initial probability distribution is referred to as \emph{strong} lumpability, while otherwise it is termed \emph{weak} lumpability \cite{buchholz_lump,Sokolova03onrelational}. Approximate aggregation techniques for Markov chains of biochemical networks are discussed in \cite{VerenaSW}.  Probabilistic bisimulation was introduced as an extension to classic bisimulation in \cite{LS89:ProbBis}. It is extended to continuous-state and continuous-time in \cite{PP:BisimulLMP02} and, for the discrete-state case, to weak bisimulation \cite{Hermanns99:weakbisimulation}. For instance, in \cite{PP:BisimulLMP02} the authors use bisimulation of labelled Markov processes, the state space of which is not necessarily discrete, and they provide a logical characterization of probabilistic bisimulation. Another notion of weak bisimulation was recently introduced in \cite{Doyen:EMC08}. 
Therein two labeled Markov chains are defined to be equivalent if every finite sequence of observations has the same probability of  occurring in the two chains. Herein we recognize the sound aggregations of \cite{journal} as a form of \emph{backward Markov bisimulations} on weighted labeled transition systems (WLTS), and we show it to be equivalent to the notion of \emph{weak lumpability} on Markov chains.    

The remaining part of the paper is organized as follows. In the Section~\ref{sec:prel}, we introduce weighted labeled transition systems (WLTS) and we assign it the trace density semantics of a continuous-time Markov chain (CTMC). Moreover, we define the Kappa language, and we assign a WLTS to a Kappa specification. Based on the notion of the annotated contact map we briefly summarize in Section~\ref{sec:reduction} the general procedure to compute stochastic fragments, as it is offered in \cite{journal}. In Section~\ref{sec:abstraction}, we introduce the characterizations of sound and complete abstractions on WLTS as a backward Markov bisimulation. Moreover, we define it being equivalent to the weak lumpability on Markov chains. Finally, we provide in Section~\ref{sec:crosstalk} results for the achieved dimensionality reduction for a rule-based model of the crosstalk between the EGF/insulin signaling pathway \cite{conzelmann2008}. This mechanistic model comprises $76$ rules giving rise to $42956$ reactions and $2768$ molecular species. 

\section{Preliminaries}
\label{sec:prel}

The stochastic semantics of a biochemical network is modelled by a continuous-time Markov chain (CTMC). The main object that we will use in the analysis is the weighted labelled transition system (WLTS) on a countable state space. We will assign a WLTS to a given Kappa specification, and we manipulate that object when reasoning about abstractions. 

\subsection{CTMC}

We will observe the CTMC that is generated by the weighted labelled transition system (WLTS) on a countable state space. We define the CTMC of a WLTS, by defining the Borel $\sigma$-algebra containing all cylinder sets of traces \cite{Kemeny1976} that can occur in the system, and the corresponding probability distribution among them. We also introduce the standard notation of a rate matrix, which we will use when analysing the lumpability and bisimulation properties in Sec.~\ref{sec:abstraction}.

\begin{definition} (WLTS) \label{dfn1}
A weighted-labelled transition system ${\cal W}$ is a tuple $(\S,\L,w,\pi_0)$, where 
\begin{myitemize}
\item $\S$ is a countable state space;
\item $\L$ is a set of labels;
\item $w:\S \times \L \times \S\ra {\RB}_0^+$ is the weighting function that maps two states and a label to a real value; 
\item $\pi_0:\S\ra [0,1]$ is an initial probability distribution.
\end{myitemize}
We assume that the label fully identifies the transition, i.e.~for any $x\in \S$ and $l\in\L$, there is at most one $x'\in \S$, such that $w(x,l,x')>0$. Moreover, we assume that the system is finitely branching, in the sense that (i) the set $\{x\in\S \;\mid\; \pi_0(x)>0\}$ is finite, and (ii) for arbitrary $\hat x\in \S$, the set $\{(l,x')\in\L\times\S\;\mid\; w(\hat x,l,x')>0\}$ is finite.

The \emph{activity} of the state $x_i$, denoted $a:\S\ra {\RB}_0^+$ is the sum of all weights originating at $x_i$, i.e.~
\begin{align*}
a(x_i): = \sum \{w(x_i, l,x_j) \;\mid\; x_j\in \S\!,\; l\in\L\} 
.\end{align*}
\end{definition}

The definition of a WLTS implicitely defines a \emph{transition relation} $\rightarrow\subseteq \S \times \S$, such that $(x_i,x_j)\in \rightarrow$, if and only if there exists a non-zero transition from state $x_i$ to state $x_j$, i.e.~the total weight over all labels is strictly bigger then zero, written $\sum \{w(x_i,l,x_j) \;\mid\; l\in \L \} > 0$. Moreover, we can differentiate the \emph{initial set of states} ${\cal I}\subseteq \S$, such that their initial probabilities are positive, i.e.~${\cal I} = \{x\in {\S} \;\mid\; \pi_0(x) > 0\}$.

\begin{definition}(Rate matrix of a WLTS) \label{dfn2}
Given a WLTS ${\cal W} = (\S,\L,w,\pi_0)$, we assign it the CTMC rate matrix $R:\S\times \S \ra {\RB}$, given by 
$R(x_i,x_j) =  \sum\{w(x_i,l,x_j) \;\mid\; {l\in \L}\}$.
\end{definition}

The consequence is that we do not enforce $R(x_i,x_i) = -\sum\{R(x_i,x_j)\;\mid\;i\neq j\}$, as it is usual for the generator matrix of CTMCs. This however does not affect the transient, not the steady-state behavior of the CTMC \cite{baier03}. We do so for the following reason. When abstracting the WLTS by partitioning the state space, we get another WLTS. If the two states $x$ and $x'$ which have a transition between each other were aggregated in the same partition class $\xAbs$, it will result as a prolongation of the residence time in the abstract state $\xAbs$, i.e.~we will have a self-loop in the abstract WLTS. 

If we refer to the generated stochastic Markov process, written as a continuous-time random variable $\{X_t\}_{t\in {\RB}_0^+}$, over the countable state space $\S$. We write $\Pr(X_t = x_i)$, the probability that the process takes the value $x_i$ at time point $t$. It thus holds that $\Pr(X_0=x_i) = \pi_0(x_i)$, and $\Pr(X_{t+\mathrm{d}t} =x_j\;\mid\;X_t=x_i)=R(x_i,x_j) \mathrm{d}t$ when $i\neq j$, whereas $\Pr(X_{t+\mathrm{d}t} =x_i\;\mid\;X_t=x_i)= R(x_i,x_i)\mathrm{d}t + (1- \sum \{R(x_i,x_{j'})\mathrm{d}t \;\mid\; {x_{j'}\in\S} \})$, which gives after simplification: $\Pr(X_{t+\mathrm{d}t} =x_i\;|\;X_t=x_i) =1- \sum \{R(x_i,x_{j'})\mathrm{d}t \;\mid\; {j'\neq i} \}$.

We define the cylinder sets of traces that can be observed in the system ${\cal W}$. By observing the trace at a certain time point, we mean observing the sequence of visited states, labels that were assigned to the executed transitions, and time points of when the transition happened.

\begin{definition} (A trace of a WLTS) \label{dfn3}
Let us observe the  WLTS ${\cal W}=(\S,\L,w,\pi_0)$ and its CTMC.
Given a number $k$ in ${\mathbb N}$,
we define a \emph{trace} of length $k$ as
$\tau\in (\S\times \L\times {\RB}_0^+)^k\times \S$, written
\begin{equation*}
\tau = \trans[t_1]{x_0}{l_1}{x_1}  \ldots \trans[t_1+...+t_k]{x_{k-1}}{l_k}{x_k}.
\end{equation*}
If the trace $\tau$ is such that 
(i) $\pi_0(x_0)>0$, and 
(ii) for all $i$, $0\leq i\le k$, we have that $w(x_i,l_i,x_{i+1}) > 0$, then we say that $\tau$ belongs to the set of traces of ${\cal W}$, and we write it 
$\tau\in {\cal T}({\cal W})$.
\end{definition}
The `time stamps' on each of the transitions denote intuitively the absolute time of the transition, from the moment when the system was started ($t=0$).
We cannot assign the probability distribution to the traces in ${\cal T}(\cal W)$, since the probability of any such trace is zero. We thus introduce the cylinder set of traces over intervals of times. 
\begin{definition}  (Cylinder set of traces) \label{dfn4}
If ${\IReal}$ is the set of all nonempty intervals in ${\RB}_0^+$, then we define the \emph{cylinder} set of traces
$\Tauu\in  (\S\times \L\times \IReal)^k\times \S$, such that:
\begin{align}
\Tauu = \trans[I_1]{x_0}{l_1}{x_1}  \ldots \trans[I_k]{x_{k-1}}{l_k}{x_k} 
\label{traceB}
\end{align}
denotes the set of all traces $\tau=\trans[t_1]{x_0}{l_1}{x_1}  \ldots \trans[t_1+...+t_k]{x_{k-1}}{l_k}{x_k}$, such that $t_i\in I_i$, $1\leq i\leq k$.
If the cylinder of traces $\Tauu$ is such that $\pi_0(x_0)>0$, and for all $i=0,...,k-1$, 
we have that $w(x_i,l_i,x_{i+1}) > 0$, then we say that $\Tauu$ belongs to the cylinder set of traces of ${\cal W}$, and we write 
$\Tauu \in \tset$. 
\end{definition}
Let $\Omega(\tset)$ be the smallest Borel $\sigma$-algebra that contains all the cylinder sets of traces in $\tset$.
We define a probability measure over $\Omega(\tset)$ in the following way.
\begin{definition} (Trace density semantics on a WLTS)  
Given a WLTS $(\S,\L,w,\pi_0)$, and a number $k$ in ${\mathbb N}$, the probability of the cylinder set of traces 
$\Tauu \in \tset$, specified as in expression ($\ref{traceB}$), 
is given by:
\begin{align*}
\pi(\Tauu) =\pi( \trans[I_1]{x_0}{l_1}{x_1}  \ldots \trans[I_k]{x_{k-1}}{l_k}{x_k} )=\ptraceTanja.
\end{align*}
\label{dfn5}
\end{definition}
Note that $\int_{I_i} a(x_{i-1})e^{-a(x_{i-1}){\cdot}t} {\mathrm{d}t}=e^{-a(x_{i-1})\cdot \inf(I_i)}-e^{-a(x_{i-1})\cdot \sup(I_i)}$ is the probability of exiting the state $x_{i-1}$ in a time interval $I_{i-1}$, since the probability density function of the residence time of $x_{i-1}$ is equal to $a(x_{i-1})e^{-a(x_{i-1})}$.

\subsection{Kappa} 
\label{sec:kappa}

We present Kappa in a process-like notation. We start with an operational semantics, then define the stochastic semantics of a Kappa model.

We assume a finite set of agent names $\agentname$, representing different kinds of proteins; a finite set of sites $\sitename$, corresponding to protein domains; a finite set of internal states $\stateset$, and $\interface_\iota{}$,$\interface_{\lambda}$ two signature maps from $\agentname$ to $\wp(\sitename)$, listing the domains of a protein which can bears respectively an internal state and a binding state. We denote by $\interface$ the signature map that associates to each agent name $A\in\agentname$ the combined interface $\interface_\iota(A)\cup\interface_\lambda(A)$.

\begin{definition} (Kappa agent)
A \emph{Kappa agent} $A(\sigma)$ is defined by its type $A\in \agentname$ and  its \emph{interface} $\sigma$. In $A(\sigma)$, the interface $\sigma$ is a sequence of sites $s$ in $\interface(A)$, with internal states (as subscript) and binding states (as superscript). The internal state of the site $s$ may be written as \site{s}{\epsilon}{}, which means that either it does not have internal states (when $s\in \Sigma(A) \setminus \Sigma_\iota(A)$), or it is not specified. A site that bears an internal state $m\in\stateset$ is written \site{s}{m}{} (in such a case $s\in\interface_\iota(A)$). The binding states of a site $s$ can be specified as \site{s}{}{\epsilon}, if it is \emph{free}, otherwise it is bound (which is possible only when $s\in\interface_\lambda(A)$). There are several levels of information about the binding partner: we use a binding label $i\in\bindinglabel$ when we know the binding partner, or a wildcard bond $\bound{}$ when we only know that the site is bound. The detailed description of the syntax of a Kappa agent is given by the following grammar: 
\begin{equation*}
\begin{array}{rcllrccll}
a & ::= & \agent{$N$}{$\sigma$} & \hbox{(agent)} & \hspace*{5mm} &  s  & ::=  & \text{\site{n}{\iota}{\lambda}} & \hbox{(site)}\cr
\agent{$N$}{} & ::= & A\in \agentname & \hbox{(agent name)}  &&\text{\site{n}{}{}} & ::=  & x\in\sitename & \hbox{(site name)} \cr
\sigma & ::= & \emptylist \;\mid\; s{\sitesep}\sigma & \hbox{(interface)} && \iota & ::= & \emptystate \;\mid\; m\in\stateset & \hbox{(internal state)}\cr
&&&&& \lambda & ::= & \emptybinding \;\mid\; \bound{}\;\mid\;i\in \bindinglabel & \hbox{(binding state)} \cr
\end{array}
\end{equation*}
We generally omit the symbol $\emptybinding$.
\end{definition}

\begin{definition} (Kappa expression)
\emph{Kappa expression} $E$ is a set of agents \agent{A}{$\sigma$} and fictitious  agents $\emptyagent$. Thus the syntax of a Kappa expression is defined as follows:
\begin{align*}
E & ::= \emptylist \;\mid\; a{\agentsep}E \;\mid\; \emptyagent{\agentsep}E.
\end{align*}
\end{definition}

The structural equivalence $\equiv$, defined as the smallest binary equivalence relation between expressions that satisfies the rules given as follows:
\begin{equation*}
 \begin{array}{rcl}
   E{\agentsep}\agent{$A$}{$\sigma{\sitesep}s{\sitesep}s'{\sitesep}\sigma'$}{\agentsep}E' & \equiv & E{\agentsep}\agent{$A$}{$\sigma{\sitesep}s'{\sitesep}s{\sitesep}\sigma'$}{\agentsep}E' \\
   E{\agentsep}a{\agentsep}a'{\agentsep}E' & \equiv & E{\agentsep}a'{\agentsep}a{\agentsep}E' \\
   E & \equiv & E{\agentsep}\emptyagent \end{array} \hfill
\begin{array}{c}
\strutfrac{\displaystyle
  i,j \in \bindinglabel \textnormal{ and }
  i \textnormal{ does not occur in } E}{
  \displaystyle E[i/j] \equiv E} \\
\strutfrac{\displaystyle 
  i \in \bindinglabel \textnormal{ and } i \textnormal{ occurs only once in } E}{  \displaystyle E[\emptybinding/i] \equiv E}
\end{array}
\end{equation*}
 stipulates that neither the order of sites in interfaces nor the order of agents in expressions matters, that a fictitious agent might as well not be there, that binding labels can be injectively renamed and that \emph{dangling bonds} can be removed.

\begin{definition} (Kappa pattern, Kappa mixture)
A Kappa pattern is a Kappa expression which satisfies the following five conditions: (i) no site name occurs more than once in a given interface; (ii) each site name $s$ in the interface of the agent $A$ occurs in $\interface(A)$; (iii) each site $s$ which occurs in the interface of the agent $A$ with a non empty internal state occurs in $\Sigma_\iota(A)$; (iv) each site $s$ which occurs in the interface of the agent $A$ with a non empty binding state occurs in $\Sigma_\lambda(A)$; and (v) each binding label $i\in\bindinglabel$ occurs exactly twice if it does at all --there are no dangling bonds. A \emph{mixture} is a pattern that is fully specified, i.e.~each agent $A$ documents its full interface $\interface(A)$, a site can only be free or tagged with a binding label $i\in\bindinglabel$, a site in $\interface_\iota(A)$ bears an internal state in $\stateset$, and no fictitious agent occurs.
\end{definition}

\begin{definition} (Kappa rule)
A Kappa rule $r$ is defined by two Kappa patterns $E_\ell$ and $E_r$, and a rate $k\in \RR_0^+$, and is written: $r = {E_\ell \rightarrow E_r @ k}$.

A rule $r$ is well-defined, if the expression $E_r$ is obtained from $E_\ell$ by finite application of the following operations: (i) creation (some fictitious agents $\emptyagent$ are replaced with some fully defined agents of the form \agent{A}{$\sigma$}, moreover $\sigma$ documents all the sites occurring in $\interface(A)$ and all site in $\interface_\iota(A)$ bears an internal state in $\stateset$), (ii) unbinding (some occurrences of the wild card and binding labels are removed), (iii) deletion (some agents with only free sites are replaced with fictitious agent $\emptyagent$), (iv) modification (some non-empty internal states are replaced with some non-empty internal states), (v) binding (some free sites are bound pair-wise by using binding labels in $\bindinglabel$). 
\end{definition}

From now on, we assume all rules to be well-defined. We sometimes omit the rate of a rule. Moreover, we denote by $E_\ell\lrar E_r @k_1,k_2$  the two rules $E_\ell \lrar E_r\rar @k_1$ and $E_r \rar E_\ell @k_2$.

\begin{definition} (Kappa system) 
A \emph{Kappa system} $\mathcal{R} = (\pi_0^{\mathcal{R}},\{r_1,\ldots,r_n\})$ is given by finite distribution over initial mixtures $\pi_0^{\mathcal{R}}\;:\;\{M_{0_1},\ldots,M_{0_k}\} \rightarrow [0,1]$, and a finite set of rules $\{r_1,\ldots,r_n\}$. 
\end{definition}

In order to apply a rule $r := E_{\ell} \rar E_{r} @k$ to a mixture $M$, we use the structural equivalence $\equiv$ to bring the participating agents to the front of $E$ (with their sites in the same order as in $E_{\ell}$), rename binding labels if necessary and introduce a fictitious agent for each agent that is created by $r$. This yields an equivalent expression $E'$ that \emph{matches} the left and side (\rlhs)  $E_\ell$, which is written ${E}\sat{E_\ell}$ as defined as follows:

{\newcommand{\lla}{\emptyagent  \sat{}  \emptyagent }
\newcommand{\llb}{\lambda_\ell \in \{\bound{},\lambda\} \implies \lambda \sat \lambda_\ell}
\newcommand{\llc}{\iota_\ell \in \{\emptystate,\iota\} \implies\iota \sat \iota_\ell}
\newcommand{\lld}{\iota \sat \iota_\ell \logicand \lambda\sat \lambda_\ell \implies \text{\site{n}{\iota}{\lambda}}  \sat  \text{\site{n}{\iota_\ell}{\lambda_\ell}}}
\newcommand{\lle}{\sigma  \sat  \emptylist}
\newcommand{\llf}{s\sat s_\ell \logicand  \sigma\sat \sigma_{\ell} \implies s\sitesep\sigma \sat s_\ell\sitesep\sigma_\ell}
\newcommand{\llg}{\sigma \sat \sigma_\ell \implies \text{\agent{$N$}{$\sigma$}} \sat \text{\agent{$N$}{$\sigma_\ell$}}}
\newcommand{\llh}{E  \sat \emptylist}
\newcommand{\lli}{a\sat  a_{\ell} \logicand  E \sat E_{\ell} \implies a{\agentsep}E \sat a_{\ell}{\agentsep}E_{\ell}}
\begin{equation*}
\begin{array}{ccc}
\llh & \llg & \lld \cr 
\hspace*{-1mm} \lli \hspace*{-3mm} \mbox{} & \lle & \llc \cr
\lla & \llf \hspace*{-4mm} \mbox{} & \llb \cr
\end{array}
\end{equation*}}

Note that in order to find a matching, we only uses structural equivalence on $E$, not $E_{\ell}$. We then replace $E'$ by $E'[E_r]$ which is defined as follows:
{%
\newcommand{\lla}{\emptyagent[a_r]  =  a_r}
\newcommand{\llb}{a_r[\emptyagent] =  \emptyagent}
\newcommand{\llc}{\lambda_r\in\bindinglabel\cup\{\emptybinding\} \implies \lambda[\lambda_r] = \lambda_r}
\newcommand{\lld}{\lambda[\bound{}]  =  \lambda}
\newcommand{\lle}{\iota_r\in\stateset \implies \iota[\iota_r]=\iota_r}
\newcommand{\llf}{\text{\site{n}{\iota}{\lambda}}[\text{\site{n}{\iota_r}{\lambda_{r}}}]  =   \text{\site{n}{\iota[\iota_r]}{\lambda[\lambda_r]}}}
\newcommand{\llg}{\sigma[\emptylist]  =     \sigma}
\newcommand{\llh}{(s\sitesep\sigma)[s_r\sitesep\sigma_r]  =  s[s_r]\sitesep\sigma[\sigma_r]}
\newcommand{\lli}{\agent{$N$}{$\sigma$}[\agent{$N$}{$\sigma_r$}]  =  \agent{$N$}{$\sigma[\sigma_{r}]$}}
\newcommand{\llj}{E[\emptylist]  =  E}
\newcommand{\llk}{(a{\agentsep}E)[a_r{\agentsep}E_r]  =  a[a_{r}]{\agentsep}E[E_{r}]}
\begin{equation*}
\begin{array}{ccc}
\llj & \lli & \llf \cr
\llk & \llg & \lle \cr
\lla & \llh & \llc  \cr
\llb & & \lld
\end{array}
\end{equation*}
This may produce dangling bonds (if $r$ unbinds a wildcard bond or destroys an agent on one side of a bond) or fictitious agents (if $r$ destroys agents), so we use $\equiv$ resolve them. 
}
\subsubsection{Population-based stochastic semantics}

In addition to the rate constants $k$, careful counting of the number of times each rule can be applied to a mixture is required to define the system's quantitative semantics correctly \cite{fmsb:2008,ecsb:2009}. Thus we define the notions of embedding between a mixture and an expression. Let $Z=a_1\agentsep\ldots\agentsep a_m$ and $Z_\ell=c_1,\ldots,c_n$ be two patterns with no occurrence of the fictitious agent and such that there exists a pattern $Z'=b_1,\ldots,b_m$ that satisfies both $Z\equiv Z'$ and $Z' \sat Z_\ell$ (and so, in particular, $n\leq m$).

The agent permutations used in the proof that $Z\equiv Z'$ allow us to derive a permutation $\texttt{p}$ such that $a_{\texttt{p}(i)}\equiv b_{i}$. The restriction $\phi$ of $\texttt{p}$ to the integers between $1$ and $n$ is called an \emph{embedding} between $Z_\ell$ and $Z$. This is written $\embed{Z_\ell}{Z}{\phi}$. There may be several embeddings between $Z_\ell$ and $Z$ for the same $Z'$; if so, this influences the relative weight of the reaction in the stochastic semantics. We denote by $\embedset{Z}{Z'}$ the set of embeddings between $Z$ and $Z'$. This notion of embedding is extended to patterns (including fictitious agents) by defining $\embed{Z_\ell}{Z}{\phi}$ if, and only if,  $\embed{(\downarrow_\emptyagent Z_\ell)}{(\downarrow_\emptyagent Z)}{\phi}$, where $\downarrow_\emptyagent$ removes all occurrences of the fictitious agent in patterns. 

We assume that $E_\ell$ is the \rlhs\ of a rule $r := E_\ell \rar  E_r @k$ and $Z$ is a mixture such that $\embed{E_\ell}{Z}{\phi}$. Let  $Z=a_1,\ldots,a_m$ and $\downarrow_\emptyagent E_\ell = c_1,\ldots c_n$. Given $Z' \equiv Z$ (we write $\downarrow_\emptyagent Z'= b_1,\ldots,b_m$) and a bijection $\texttt{p}$ such that we have $Z' \sat E_\ell$, $b_i\equiv a_{\texttt{p}(i)}$ for $1\leq i \leq m$ and $\phi(j)=\texttt{p}(j)$ for $1 \leq j \leq n$. The result of applying $r$ along $\phi$ to the mixture $Z$ is defined (modulo $\equiv$) as any pattern that is  $\equiv$-equivalent to $Z'[E_r]$. In other words the embedding $\phi$ between $E_\ell$ and $Z$ fully defines the action of $r$ on $Z$ up to structural equivalence.

We are now ready to define the stochastic semantics by the mean of a WLTS. In this semantics, the state is a soup of agents, that is to say that we do not care about the order of agents in mixture. So the states of the system are the class of $\equiv$-mixture. 

Defining species as connected mixture, the state of the system can be seen as a multi-set of species. The formal definition of a Kappa species is as follows: 
\begin{definition} (Kappa species)
A pattern $E$ is \emph{reducible} whenever $E\equiv E',E''$ for some non-empty patterns $E'$, $E''$; A \emph{Kappa species} is the $\equiv$-equivalence class of a irreducible Kappa mixture. 
\end{definition}

As explained earlier, the action of a rule $r$ on a mixture $E$ is fully defined (up to $\equiv$) by an embedding $\phi$ between the \rlhs\ $E_\ell$ of the rule $r$ and the mixture.  So as to consider computation steps over $\equiv$-equivalent of mixtures,  we introduce an  equivalence relation $\equiv_{\mathcal{L}}$ over triples $(r,E,\phi)$ where $\phi$ is an embedding of the \rlhs\ $E_\ell$ of $r$ into $E$. We say that $(r_1,E_1,\phi_1) \equiv_{\mathcal{L}} (r_2,E_2,\phi_2)$ if, and only if, (i) $r_1=r_2$ and (ii) there exists an embedding $\psi\in\embedset{E_1}{E_2}$ such that $\phi_2=\psi\circ\phi_1$.

\begin{definition} (WLTS of a Kappa system) \label{dfn:WLTSKappa}
\label{dfn9}
Let $\mathcal{R}=(\pi_0^\mathcal{R},\{r_1,\ldots,r_n\})$ be a Kappa system. We define the WLTS $\mathcal{W}_\mathcal{R}=(\mathcal{X},\mathcal{L},w,\pi_0)$ where:
(i) $\mathcal{X}$ is the set of all $\equiv$-equivalent classes of mixtures;
(ii) $\mathcal{L}$ is the set of all $\equiv_{\mathcal{L}}$-equivalence classes of triples $(r,E,\phi)$ such that $\phi$ is an embedding between the \rlhs\ $E_\ell$ of $r$ and $E$;
(iii) $w(x,l,x') = \strutfrac{k}{\nauto{E_\ell}}$ whenever there exist a rule $r=E_\ell \rar E_r@k$, two mixtures $E$ and $E'$, and an embedding $\phi\in\embedset{E_\ell}{E}$, such that $x=[E]_{\equiv}$, $l=[r,E,\phi]_{\equiv_{\mathcal{L}}}$, and $E'$ is the result (up to $\equiv$) of the application of $r$ along $\phi$ to the mixture $E$; otherwise $w(x,l,x')=0$;
(iv) $\pi_0(x) = \sum\{ \pi^{\mathcal{R}}_0(E')\;\mid\; {E'\in\textit{Dom}(\pi_0^{\mathcal{R}})\cap x}\}$.
\end{definition}

The stochastic semantics of a Kappa system $\mathcal{R}$ is then defined as the trace distribution semantics of the WLTS $\mathcal{W}_{\mathcal{R}}$. 

\section{Reduction procedure} 
\label{sec:reduction}

\newcommand{\test}[3]{{\bf{test}}_{#1}^{#2}(#3)}
\newcommand{\transposition}[2]{[{}^{#1}_{#2}]}

In this section, we assume, without any loss of generality that $\interface_\iota$ and $\interface_\lambda$ are disjoint sets. This can always be achieved by taking two disjoint copies $\sitename_\iota$ and $\sitename_\lambda$ of $\sitename$ and using site names in $\sitename_\iota$ to bear internal states, and site names in $\sitename_\lambda$ to bear binding states. 

Informally, a contact map represents a summary of the agents and their potential bindings. 
\begin{definition} (Contact map) 
Given a Kappa system ${\cal R}$, a \emph{contact map} (CM) is a graph object $({\cal N}, {\cal E})$, where the set of nodes ${\cal N}$ are agent types equipped with the corresponding interface, and the edges are specified between the sites of the nodes. Formally, we have that ${\cal N} =\{(A,\Sigma(A)) \;\mid\; A\in {\cal A}\}$ and ${\cal E} \subseteq \{((A,s),(A',s')) \;\mid\; A,A'\in {\cal A} \hbox{ and } s\in \Sigma(A), s'\in \Sigma(A')\}$. If the site $s$ of an agent of type $A$ and the site $s'$ of an agent of type $A'$ bear the same binding state in the rhs $E_r$ of a rule, then there exists an edge $e=((A,s),(A',s'))\in {\cal E}$ between  $s\in \Sigma(A)$ and $s'\in \Sigma(A')$.
\end{definition}

We say that a site $s$ of the agent $a$ is \emph{tested} by the rule $r$, if it is contained in the \rlhs\ $E_\ell$ of the rule $r$. 
 
\begin{definition} (Annotated Contact map)\label{dfn:ACM}
Given a Kappa system ${\cal R}$, and its CM $(\cal N, \cal E)$, a valid annotated contact map (ACM) $({\mathscr N},\cal E)$ is a contact map where all agents  are annotated with respect to the rule set ${\cal R}$. The \emph{annotation} on the agent of type $A\in {\cal A}$ with respect to the rule $r$ is given by the equivalence relation on its set of sites $\approx_A\subseteq \Sigma(A)\times \Sigma(A)$ such that:
\begin{myitemize} 
\item If a rule $r$ tests the sites $s_1$ and site $s_2$ of agents $a_1$,$a_2$ (it is possible that $a_1=a_2$) of type $A$, then $s_1\approx_A s_2$;
\item If a rule $r$ creates an agent $a$ of type $A$, then all the sites of $\Sigma(A)$ are in the same equivalence  class, i.e.~$\approx_A = \Sigma(A)\times \Sigma(A)$;
\end{myitemize} 
Note that there can be several annotations of the agent type $A\in\agentname$ which satisfy the conditions. More precisely, if the equivalence relation  $\approx_A$ meets the condition, then so does any of its refinement. This allows to define the smallest such equivalence relation $\approx_A$  which we call the \emph{minimal} annotation of agent $A$. An ACM is \emph{minimal} whenever each agent type is annotated by its minimal annotation. 
\end{definition}

Let $r$ be a rule and an ACM which is valid with respect to the singleton $\{r\}$. Then for any agent type $A\in\agentname$, either $A$ does not occur in the \rlhs\ of $r$, or $A$ occurs but all occurrences of $A$ have an empty interface, or $A$ occurs, tests some sites which all belongs to a same equivalence  class $C$ in $\approx_A$. In the latter case, we define $\test{r}{\textit{ACM}}{A} = C$, otherwise, we define $\test{r}{\textit{ACM}}{A}=\emptyset$. 

The meaning of the ACM is to summarize the dependences between sites that can occur during the simulation of a Kappa system. If the two sites $s$ and $s'$ in the $\interface(A)$ are correlated by the relation $\approx_A$, i.e.~$s\approx_A s'$, it suggests that they are \emph{dependent} in the following way. We must not aggregate in the same equivalence  class any two states $x$ and $x'$, such that they contain the agent $A$ in a different evaluation of the sites $s$ and $s'$. On the other hand, if the two sites $s$ and $s'$ are not correlated by $\approx_A$, then we may aggregate the states by the 'marginal' criteria, i.e.~the condition which involves only one of the sites.  Therefore, the less states are related by $(\approx_A)_{A\in\agentname}$, the better the reduction will be. To numerically justify this, we can imagine having an agent type $A$ whose interface has $n$ different sites $s_1,...,s_n$, and each of them has two possible internal state modifications. Let us observe the two limiting relations $\approx_A$, i.e.~$\approx_A=\{(s_i,s_j) \;\mid\; 1\leq i \leq n,1\leq j \leq n\}$, and $\approx_A'=\{(s_i,s_i) \;\mid\; 1\leq i\leq n\}$. The annotation $\approx_A$ enforces at least to $2^n$ states to describe all modifications of the agent $A$, whereas the annotation $\approx_A'$ suggests that it is enough to use only $2\cdot n$ of them.

The ACM can be used to identify parts of Kappa species that we call fragments. 
\begin{definition} (Kappa fragments)
\label{dfn14}
A fragment is the $\equiv$-equivalent class of a non empty irreducible pattern $E$ such that: (i) the set of sites in the interface $\sigma$ of an agent \agent{A}{$\sigma$} in $E$ is an equivalence  class of $\approx_A$, (ii) sites can only be free or tagged with a binding label $i\in\mathbb{N}$ and sites in $\Sigma_\iota$ are tagged with an internal state in $\stateset$, (iii) there is no occurrence of fictitious agent $\emptyagent$.
\end{definition}

We can use fragments to abstract the WLTS $\mathcal{W}_\mathcal{R}$, by identifying the mixtures which  have the same (multi-)set of fragments.  To reach that goal, we first overload the definition of $\equiv$ in order to identify mixtures having the same fragments. We introduce the binary relation $\equiv^\sharp$ as the smallest equivalence relation over patterns which is compatible with $\equiv$ and such that:
\begin{align*}
A(\sigma),A(\sigma'),E \equiv^\sharp  A(\uparrow_C\sigma',\uparrow_{\interface(A)\setminus C}\sigma),A(\uparrow_C\sigma,\uparrow_{\interface(A)\setminus C} \sigma'),E
\end{align*}
for any agent type $A\in\agentname$, $\sigma$,$\sigma'$ interfaces, $E$ pattern, and $C$ an $\approx_A$-equivalence class of sites.  For any set of sites $X\subseteq \sitename$, the projection function $\uparrow_X$ over interfaces keeps only the sites in $X$, formally $\uparrow_X$ is defined by $\uparrow_X \emptylist = \emptylist$, $\uparrow_X(\text{\site{s}{\iota}{\lambda}\sitesep} \sigma') = \text{\site{s}{\iota}{\lambda}\sitesep} \uparrow_X \sigma'$ whenever $s\in X$, and $\uparrow_X(\text{\site{s}{\iota}{\lambda}\sitesep} \sigma') = \uparrow_X \sigma'$ otherwise.

 Now we define the relation $\sim_{\mathcal{L}^\sharp}$ which stipulates that  the rule $r_1$ applies on $E_1$ along $\phi_1$ the same way as the rule $r_2$ on $E_2$ along $\phi_2$. More formally, we write  $(r_1,E_1,\phi_1) \sim_{\mathcal{L}^\sharp} (r_2,E_2,\phi_2)$ whenever the following properties are all satisfied:
\begin{myenumerate}
\item $r_1=r_2$;
\item $E_1\equiv^\sharp E_2$;
\item $\phi_2 = \psi\circ\phi_1$, where $\psi$ is the  permutation which tracks how the sub-interface $\uparrow_{\test{r}{\textit{ACM}}{A_i}}(A_i(\sigma_i))$ is moved in the proof that $E_1 \equiv^\sharp E_2$, for any agent $A_i(\sigma_i)$ occurring in $E_1$.

More precisely, the transposition $\transposition{i}{i+1}$ is associated to an agent permutation of the agents at position $i$ and $i+1$; the transposition $\transposition{1}{2}$ is associated to a step which permutes the sub-interface $\test{r}{\textit{ACM}}{A}$ of two agents of type $A$, for any agent type $A\in\agentname$; any other step is associated with the identity function (over $\mathbb{N}$). The function $\psi$ is defined as the composition of all the permutations (in the reverse order) which are associated to the elementary steps in the proof that $E_1 \equiv^\sharp E_2$.

\item the result of the application of $r_1$ to $E_1$ along $\phi_1$ is $\equiv$-equivalent to the result of the application of $r_2$ to $E_2$ along $\phi_2$. 
\end{myenumerate}

\begin{definition} (Abstract WLTS of a Kappa system)\label{dfn:WLTSKappaAbs}
\label{dfn13}
Let $\mathcal{R}=(\pi_0^\mathcal{R},\{r_1,\ldots,r_n\})$ be a Kappa system. We define the WLTS $\tilde{\mathcal{W}_\mathcal{R}}=(\tilde{\mathcal{X}},\tilde{\mathcal{L}},\tilde{w},\tilde{\pi_0})$ where:
\begin{myitemize}
\item $\tilde{\mathcal{X}}$ is the set of all $\equiv^\sharp$-equivalent class of mixture;
\item $\tilde{\mathcal{L}}$ is the set of all $\equiv^\sharp_{\mathcal{L}}$-equivalent class of triples $(r,E,\phi)$  such that $\phi$ is an embedding between the \rlhs\ $E_\ell$ of $r$ and $E$;
\item $\tilde{w}(\tilde{x},\tilde{\lambda},\tilde{x'})$ is equal to $\strutfrac{k}{\nauto{E_\ell}}$ whenever there exist a mixture $E$, a rule $r$, and an embeddings $\phi$ such that $\tilde{x} = [E]_{\equiv^\sharp}$ and $\tilde{\lambda} = [r,E,\phi]_{\equiv^\sharp_{\mathcal{L}}}$; otherwise, $\tilde{w}(\tilde{x},\tilde{\lambda},\tilde{x'})$ is equal to $0$;
\item for any $\tilde{x}\in\tilde{\mathcal{X}}$, $\tilde{\pi_0}(\tilde{x})= \sum_{E' \in \textit{Dom}(\pi_0^\mathcal{R})\cap \tilde{x}}  \pi^{\mathcal{R}}_0(E')$.
\end{myitemize}
\end{definition}

Let us define the relation $\sim$ over $\mathcal{X}$ by $[E_1]_\equiv \sim [E_2]_\equiv$ if, and only if, $E_1\equiv^\sharp E_2$ and the relation $\sim_{\mathcal{L}}$ over $\mathcal{L}$ by $[\lambda_1]_{\equiv_{\mathcal{L}}} \sim_{\mathcal{L}} [\lambda_2]_{\equiv_{\mathcal{L}}}$ if, and only if, $\lambda_1\equiv^\sharp_{\mathcal{L}} \lambda_2$.

\section{Abstraction}
\label{sec:abstraction}

We introduce abstractions on WLTS by aggregating the states and labels into partition classes.  We obtain a new WLTS defined over the aggregated states and labels. Each non-trivial abstraction is a loss of information. However some of them are such that it is possible to do the stochastic analysis on the aggregates rather than on concrete states.  We address the problem of characterizing when this is possible, and if so, how the weights in the abstracted system are computed. We also discuss the reverse process - given the abstracted system, and a particular probability distributions over the aggregates, whether we can make conclusions about the traces in the concrete system. We do the general theoretical analysis of the abstractions on WLTS, and afterwards we show the relation with the reduction of Kappa systems, that is presented in Sec.~\ref{sec:reduction}.

\begin{definition}  (Abstraction) \label{dfn12}
Consider a WLTS ${\cal W} =(\S,\L,w,\pi_0)$, and a pair of equivalence relations $(\sim,\sim_{\mathcal{L}})\in \S^2\times \L^2$, such that each $\sim$-equivalence class and each $\sim_{\mathcal{L}}$-equivalence class is finite. We denote the equivalence classes by $\tilde{x}$, $\tilde{l}$, and we write $x\in \tilde{x}$, to indicate that $x$ belongs to the equivalence class $\tilde{x}$, and $l\in \tilde{l}$ to indicate that the label $l$ belongs to the equivalence  class $\tilde{l}$.

A  WLTS of the form $\tilde{{\cal W}} = (\S_{/\sim}, \L_{/\sim_{\mathcal{L}}}, \tilde{w},\tilde{\pi_0})$, where $\tilde{\pi}_0(\tilde{x})=\sum\{\pi_0(x)\;\mid\;x\in \tilde{x}\}$ is called  an \emph{abstraction} of ${\cal W}$, induced by the pair of equivalence relations $(\sim,\sim_{\mathcal{L}})$. Note that more abstractions can be induced by ${\cal W}$, depending on how $\tilde{w}$ is defined.

Moreover, for any two cylinder sets of traces $\tildtauu\in \tildetset$ and $\Tauu\in  \tset$, we say that $\tildtauu =  \trans[I_1]{\tilde{x}_0}{\tilde{l}_1}{\tilde{x}_1}  \ldots \trans[I_k]{\tilde{x}_{k-1}}{\tilde{l}_k}\tilde{x}_k$ is an abstraction of $\Tauu =  \trans[I_1]{x_0}{l_1}{x_1}  \ldots \trans[I_k]{x_{k-1}}{l_k}{x_k}$, and we write it $\Tauu \in \tildtauu$.
\end{definition}

\begin{definition} (Sound abstraction: Aggregation) \label{dfn:sound}
We say that the abstraction $\tilde{{\cal W}}$ is a \emph{sound abstraction} of ${\cal W}$, if the probability of any cylinder set of traces $\tildtauu \in \tildetset$ is equal to the sum of the probabilities of all the cylinder sets of traces $\Tauu \in \tset$, whose abstraction is $\tildtauu$: \[\pi(\tildtauu) = \sum\{\pi(\Tauu)\;|\;\Tauu\in \tildtauu\}.\]
\label{thmsound}
\end{definition}

We introduce a function $\ga:\S_{/\sim}\ra (\S\ra [0,1])$ which assigns to each of the partition class $\tilde{x}\in \S_{/\sim}$ a probability distribution over the states that belong to this partition class. The set of all such vectors $\gamma$ we denote by ${\Gamma}_{\S,{\sim}}$ is defined as: \[\{\ga\;|\;\ga:\S_{/\sim} \ra (\S\ra [0,1])  \logicand \forall \tilde{x}\in\tilde{\mathcal{X}}, \sum_{x\in \tilde{x}} \ga(\tilde{x},x)= 1\}.\]
We can think of the value $\ga(\xAbs,x)$ as the \emph{conditional probability} of being in the state $x$, knowing that we are in state $\xAbs$, i.e.~$\Pr(X_t=x \;\mid\; X_t\in \xAbs)=\ga(\xAbs,x)$. We note that, when thinking of $\ga$ as the conditional probability, it should be a time-dependent value. However, we refer to $\ga$ as to a single, constant distribution. This will be justified in Lem.~\ref{lem1}.
\begin{definition} (Complete abstraction: Deaggregation) \label{dfn:complete}
We say that the abstraction $\tilde{{\cal W}}$ is a \emph{complete abstraction} of ${\cal W}$ for $\ga\in \Gamma_{\S,{\sim}}$, if the following holds. Given the probability of an arbitrary abstract cylinder set of traces of length $k\geq 1$, that ends in the abstract state $\tilde{x}_k$ (written $\trans[]{\tildtauu}{}{\tilde{x}_k}$), we can recompute the probability of ending the trace in the concrete state $x_k\in \tilde{x}_k$ in the following way: 
\begin{align*}
\pi(\trans[]{\tildtauu}{}{x_k}) = \ga(\tilde{x}_k,x_k)\cdot 
\pi(\trans[]{\tildtauu}{}{\tilde{x}_k}).
\end{align*}
\end{definition}
Sound and complete abstractions ${\tilde{\cal W}}$ cannot be induced by any pair of relations $(\sim,\sim_{\mathcal{L}})$, because there might not  exist a weighting function $\tilde{w}:\S_{/\sim}\times \L_{/\sim_{\mathcal{L}}}\times \S_{/\sim}\ra {\mathbb R}$, such that the conditions from Dfn.~\ref{dfn:sound} and Dfn.~\ref{dfn:complete} are met. Moreover, even if such $\wAbs$ exists, the remaining question is whether the information on the abstract system is enough to compute them.

We now restate the main Theorem from \cite{journal}, that the abstractions for Kappa systems, that we resumed in Sec.~\ref{sec:reduction}, are sound and complete.

\begin{thm} (The abstraction induced by the ACM is sound and complete) \label{thm:main} Given a Kappa system ${\cal R} = (\pi_0^{\cal R},\{r_1,...,r_n\})$, and an ACM $(\mathscr N,\cal E)$, an abstraction $\tilde{\cal W_\mathcal{R}} = (\S_{/\sim}, {\cal L}_{/_{\sim}}, \tilde{w}, \tilde{\pi_0})$ induced by the partition classes $(\sim,\sim_{\cal L})\subseteq \S^2\times \L^2$, as proposed in the Def.~\ref{dfn13} is a sound and complete abstraction of the ${\cal W}_{\mathcal{R}} = (\S,{\cal L},w,\pi_0)$, provided that for any two mixtures $M$ and $M'$ such that $M \equiv^\sharp M'$, we have: \[\pi_0([M]_{\equiv})\cdot \nauto{M'}  = \pi_0([M']_{\equiv})\cdot \nauto{M}.\] 
We consider a mixture $M$. We denote by $x\in\mathcal{X}$ the equivalence class $[M]_\equiv$, and by $\tilde{x}\in\tilde{\mathcal{X}}$ the equivalence class $[M]_{\equiv^\sharp}= [x]_{\sim}$. The conditional probability $\ga(\tilde{x},x)$ is computed as the ratio of the number of automorphisms of $x$ (embedding between $x$ and $x$) and the sum of the number of automorphisms of any $\sim$-equivalent state. Thus we have:
\[\ga(\tilde{x},x) = \frac{\nauto{x}}{\sum \{\nauto{x'}\;\mid\;x\sim x'\}}.\]
\end{thm}
The reader can find the detailed proof in \cite{journal}.

\subsection{Lumpability} 
Now we define different versions of lumpability and investigate the relationship with sound and complete abstractions.

\begin{definition} (Lumped process)  
Given a WLTS ${\cal W} = (\S,\L,w,\pi_0)$, where $\S=\{x_1,x_2,...\}$, and a partition $\sim\subseteq \S\times\S$ on its state space, we observe the continuous-time stochastic process $\{X_t\}_{t\in \RR_0^+}$, that is generated by ${\cal W}$ (Dfn.~\ref{dfn2}). We define the \emph{lumped process} $\{Y_t\}$ on the state space $\S_{/\sim} =\{\xAbs_1,\xAbs_2,...\}$ (denoted by capital indices, i.e.~$\tilde{x}_I$, $\tilde{x}_J$) and with initial distribution $\tilde{\pi_0}$, so that
\[\Pr(Y_t = \xAbs_J\;\mid\;Y_0=\xAbs_0) =\Pr(X_t\in \xAbs_J\;\mid\;X_0\in \xAbs_0).\]
\end{definition}

The lumped process is not necessarily a Markov process. 

\begin{definition} (Lumpability) \label{dfn:lump}
Given a WLTS ${\cal W} =  (\S,\L,w,\pi_0)$ that generates the process $\{X_t\}$, we say that it is \emph{lumpable} with respect to the equivalence relation  $\sim\subseteq \S\times \S$ if and only if its lumped process $\{Y_t\}$ has the Markov property.
\end{definition}

The evolution of a process depends on the initial distribution, and so does the lumpability property. We thus define the set of initial distributions, for which the lumpability holds. We denote the set of all probability distributions over $\S$ as ${\PI}_{\S}$:
\begin{align*}
\PI_{\S}=\{\pi\;\mid\;\pi:\S\ra [0,1] \hbox{ and } \sum_{x_i\in{\S}}\pi(x_i)=1\}.
\end{align*}
Moreover, we denote the set of initial distributions that produce a chain lumpable with respect to the given equivalence relation $\sim$ by $\PI^{I}_{\S,\sim}$:
\begin{align*}
\PI^{I}_{\S,\sim} := \{\pi\;\mid\;\hbox{ the lumped process initialized with } \pi 
                         \hbox{ is lumpable with respect to } \sim \}.
\end{align*}

Whenever a distribution $\pi\in \PI_{\S}$ is positive on the equivalence  class $\xAbs$, i.e.~$\sum\{\pi(x) \;\mid\; x\in \xAbs\} >0$, we denote by $\pi|_{\xAbs}(x)$, the conditional distribution over the states of $\xAbs$: $\pi|_{\xAbs}(x) = {\pi(x)}/{\pi(\xAbs)}$, when $x\in \xAbs$, and $\pi|_{\xAbs}(x) =0$, otherwise.

\begin{definition} (Strong and weak lumpability) \label{def:lumpSW}
Given a WLTS ${\cal W} =  (\S,\L,w,\pi_0)$ that generates the process $\{X_t\}$, and an equivalence relation  $\sim\subseteq \S\times\S$, 
we say that  $\{X_t\}$ is:
\begin{myitemize}
\item \textit{strongly lumpable} with respect to $\sim$, if the lumped process $\{Y_t\}$ is Markov with respect to \emph{any} initial distribution, i.e.~$\PI^{I}_{\S,\sim}   = \PI_{\S}$;
\item \textit{weakly lumpable} with respect to $\sim$, if there \emph{exists} an initial distribution that makes the lumped process $\{Y_t\}$ Markov, i.e.~$\PI^{I}_{\S,\sim}  \neq \emptyset $.
\end{myitemize}
\end{definition}

Note that the definitions of strong and weak lumpability involve the quantifiers "for all" and "exists" over the probability distributions over a set of states. Thus, checking for either of them involves in general an infinite number of checks. People have given sufficient conditions of strong and weak lumpability on discrete-time Markov chains (DTMC's) \cite{KS60:lump, weakLumpDTMC}. The results had been extended to the continuous-time case \cite{buchholz_lump, weakLumpCTMC}. We rephrase the sufficient conditions stated therein. 


In order to understand the sense of the weak lumpability characterization, we discuss the meaning of $\ga$. We recall the semantics of a WLTS ${\cal W}$ by observing the cylinder sets of traces, i.e.~$\Tauu = \trans[I_1]{x_0}{l_1}{x_1}  \ldots \trans[I_k]{x_{k-1}}{l_k}{x_k}\in \tset$. The abstraction $\tilde{\cal W}$ of ${\cal W}$, induced by $(\sim,\sim_L)$ generates an abstract cylinder set of traces, denoted  $\tildtauu = \trans[I_1]{\xAbs_0}{l_1}{\xAbs_1}  \ldots \trans[I_k]{\xAbs_{k-1}}{l_k}{\xAbs_k}\in \tildetset$.

For any cylinder set of traces $\tildtauu\in\tildetset$, we denote by $\ga_{\tildtauu}$ the distribution of the conditional probabilities over the lumped state $\xAbs_{k}$, knowing that the abstract cylinder of traces $\tildtauu$,  which ends in the abstract state $\xAbs_k$,  was observed, i.e.~\[
\ga_{\tildtauu}(x_k) = \frac{\pi(\trans[]{\tildtauu}{}{x_k})}{\pi({\tildtauu})}.\]
The definition of the complete abstraction suggests that, if $\ga_{\tildtauu}$ was independent of the traces on which it is conditioned, i.e.~$\tildtauu$, then the completeness would hold.

\begin{thm}  (Lumpability on CTMCs)  
\label{thm2}
Let us observe a WLTS ${\cal W} = (\S,\L,w,\pi_0)$ that generates the process $\{X_t\}$, and an equivalence relation  $\sim\subseteq \S\times \S$. We consider the rate matrix $\R:\S\times \S \ra \RR$. If the lumped process is Markov, then we denote its rate matrix by $\tilde{\R}:\S_{/\sim}\times \S_{/\sim}\ra \RR$. Then we have the following characterizations about the lumped process $\{\tilde{X}_t\}$:
\begin{myitemize}
\item If for all $x_{i_1},x_{i_2}\in \S$ such that $x_{i_1}\sim x_{i_2}$, and for all $\xAbs_J\in \S_{/\sim}$, we have that 
\begin{align}
\sum_{x_j\in \xAbs_J}R(x_{i_1},x_j) = \sum_{x_j\in \xAbs_J}R(x_{i_2},x_j), 
\label{q}
\end{align}
then $\{X_t\}$ is strongly lumpable with respect to $\sim$; 
We have: 
\[\tilde{R}(\xAbs_I,\xAbs_J)=  \sum\{R(x_{i_1},x_{j}) \;\mid\; {x_j\in \xAbs_J}\};\]
\item If there exists a family of probability distributions over the lumped states, $\ga\in \Gamma_{\S,{\sim}}$, such that for all $x_{j_1},x_{j_2}\in \S$ such that $x_{j_1}\sim x_{j_2}$ and for all $\xAbs_I\in \S_{/\sim}$, we have that 
\begin{align}
a(x_{j_1}) & = a(x_{j_2}) \hbox{ and } \frac{\sum_{x_i\in \xAbs_i}R(x_i,x_{j_1})}{\ga(\xAbs_J,x_{j_1})}  = \frac{\sum_{x_i\in \xAbs_I}R(x_i,x_{j_2})}{\ga(\xAbs_J,x_{j_2})},
\label{q1}
\end{align}
then 
\begin{myenumerate}
\item  If the distribution $\ga$ is in accordance with $\pi_0$, i.e.~$\pi_0|_{\S_{/\sim}}  = \ga $, then for any finite sequence of states $(x_0,\ldots,x_k)\in \mathcal{X}^{k+1}$ and any sequence of time intervals $(I_1,\ldots,I_k)\in\mathbb{IR}^k$, we consider the set $\tildtauu$ of the traces of the form $\trans[t_1]{x'_0}{l_1}{x'_1}  \ldots \trans[t_1+...+t_k]{x'_{k-1}}{l_k}{x'_k}$. For all $i$, $0\leq i \leq k$ and $x_i\sim x'_i$, and for all $i$, $1\leq i \leq k$, $t_i\in I_i$ and $l_i\in\mathcal{L}$, we have that: $\hbox{if }\pi(\tildtauu)>0 \hbox{ then } \ga_{\tildtauu} = \ga$.

In other words, knowing that we are in state $\xAbs_{I}$, the conditional probability of being in state $x\in \xAbs_{I}$ is invariant of time. 
\item The process $\{X_t\}$ is weakly lumpable with respect to $\sim$. Moreover, we have: 
\[\tilde{R}(\xAbs_I,\xAbs_J) =  \frac{\sum\{R(x_i,x_{j_2}) \;\mid\; {x_i\in \xAbs_I}\}}{\ga(\xAbs_J,x_{j_2})};\]
\end{myenumerate}
\end{myitemize} 
\end{thm}


One shall notice that Thm.~\ref{thm2} gives a weaker condition than the completeness of WLTS abstraction (eg~see Dfn.~\ref{dfn:complete}).  The main reason is that we  do not 'track' transition labels, in the sense that we observe the abstraction on the cylinder sets of  traces induced only by $\sim$, and not also by $\sim_L$. Yet, in the particular case when states fully define the transition labels (ie, if $w(x_1,l_1,x'_1)>0$,  $w(x_2,l_2,x'_2)>0$, $x_1\sim x_2$, and $x'_1\sim x'_2$, then $l_1 \sim_{\mathcal{L}} l_2$), the given condition for weak lumpability coincides with the definition of the complete abstraction of WLTS. 

The characterization of weak (resp.~strong) lumpability given in Thm.~\ref{thm2} is sufficient, but not a necessary condition: there exist systems which are strongly or weakly lumpable, but do not satisfy the conditions given in the theorem. Interestingly, there are systems, such that the characterization from Thm.~\ref{thm2} would detect as strong, but not weakly lumpable, which is counter-intuitive with the terminology. One shall also notice that the conditions of Thm.~\ref{thm2} imply that: in order to aggregate two states in the CTMC, they must not have different waiting times until the next transition (e.g.~they should have the same activity). It is stated explicitly in the characterization of weak lumpability and it can be obtained by summation over the outgoing class in the characterization of strong lumpability.

We consider a WLTS ${\cal W} = (\S,\L,w,\pi_0)$, and the set of all equivalence relations $\sim$ on $\S$, denoted $PT_{\S}$. We introduce the subsets of $PT_{\S}$, denoted $PS$, $PW$, $CS$, $CW$ in the following meaning: (i) $PS$ -the set of all equivalence relations such that $\{X_t\}$ is strongly lumpable with respect to $\sim$; (ii) $PW$ - the set of all equivalence relations such that $\{X_t\}$  is weakly lumpable with respect to $\sim$; (iii) $CS$ - the set of all equivalence relations such that $\{X_t\}$  satisfies the condition for strong lumpability given in the Thm.~\ref{thm2}; (iv) $CW$ - the set of all equivalence relations such that $\{X_t\}$   satisfies the condition for weak lumpability given in the Thm.~\ref{thm2}. 

\begin{lemma} \label{lem1} (Relations on lumpability properties and conditions)
Consider an arbitrary WLTS ${\cal W} = (\S,\L,w,\pi_0)$ and the equivalence relation  $\sim\subseteq \S\times\S$. 
We have the following relations: (1a) if $\sim\in PS$ then $\sim \in PW$, if $\sim \in CS$ then $\sim \in PS$; and if $\sim \in CW$ then $\sim \in PW$;  The converse implication does not hold for any of the statements; (2a) If $\sim \in CW$, that does not imply $\sim \in CS$; (2b) If $\sim \in CS$, that does not imply $\sim\in CW$.
\end{lemma}

\begin{proof}
The statement in the part (1) trivially follows from the Dfn.~\ref{dfn:lump} and Thm.~\ref{thm2}. To show (2a) and (2b), we consider the WLTS ${\cal W}$ specified in the Fig.~\ref{pic1}, with the state space $\S=\{x,y_1,y_2,y_3,z_1,z_2,z_3\}$. Let $\sim_1$ be an equivalence relation on $\S$, such that $y_1 \sim_1 y_2$ and $z_1 \sim_1 z_2$. By lumping the states by $\sim_1$, we get the system ${\WAbs}_1$, as shown in Fig.~\ref{pic2}. It is easy to check that $\sim_1\in CS$.
Moreover, we have that $\sim_1\in CW$, since for 
\[\ga =\left( 
\begin{array}{ccccc}
x & y_{12} & y_3 & z_{12} & z_3 \\
1 & (0.5,0.5) & 1 & (0.5,0.5) & 1  
\end{array} 
\right)
\]
the weak lumpability condition is satisfied, so $\sim_1\in CS\cap CW$. We further lump the states $y_{12}$ and $y_3$, by taking the transitive closure of the relation $\sim_1$ union $(y_1,y_3)$, denoted $\sim_2 =tc(\sim_1\cup (y_1,y_3))$ (Fig.~\ref{pic3}). This lumping is such that $\sim_2\notin CS$ because we have
\begin{align*}
y_1\sim y_3, \hbox{ but } w(y_1,l,z_{12}) > 0, \hbox{ and } w(y_3,l,z_{12})=0.
\end{align*}
On the other hand,  for \[ \ga= \left(
\begin{array}{cccc}
x & y_{123} & z_{12} & z_3 \\
1 & (1/3,1/3,1/3) & (0.5,0.5) & 1  
\end{array} 
\right)
\]
we argue that $\sim_2\in CW$. Therefore, if the initial distribution is in accordance with $\ga$, the abstraction $\WAbs_2$ is sound and complete.

If we rather lump $z_1$ and $z_2$, by $\sim_3=tc(\sim_1\cup (z_1,z_3))$, we get the system $\WAbs_3$ (Fig.~\ref{pic4}). This system is such that $\sim_3\in CS\setminus CW$. More precisely, we cannot find a $\ga$ which would witness $\sim_3\in CW$:  if such a $\ga$ existed, we would  have $\ga(\{x\})(x) = 1$, and consequently $\ga(y_{12}) = (0.5,0.5)$, and $\ga(y_3) = 1$. This implies that the conditional distribution $\ga(z_{123})$ cannot be invariant of time - it will alternate between the distributions $(0,0,1)$ and $(0.5,0.5,0)$, depending on the choice made in $x$. Note that, since $\sim_3\in CS$, it follows that $\sim_3\in PS$, and this implies $\sim_3\in PW$. 
\end{proof}
\begin{figure}[t]
\subfigure[${\cal W}$: the concrete system]%
{\begin{minipage}{0.49\linewidth}
\label{pic1}
\hspace*{0cm}\mbox{\includegraphics[scale=0.2]{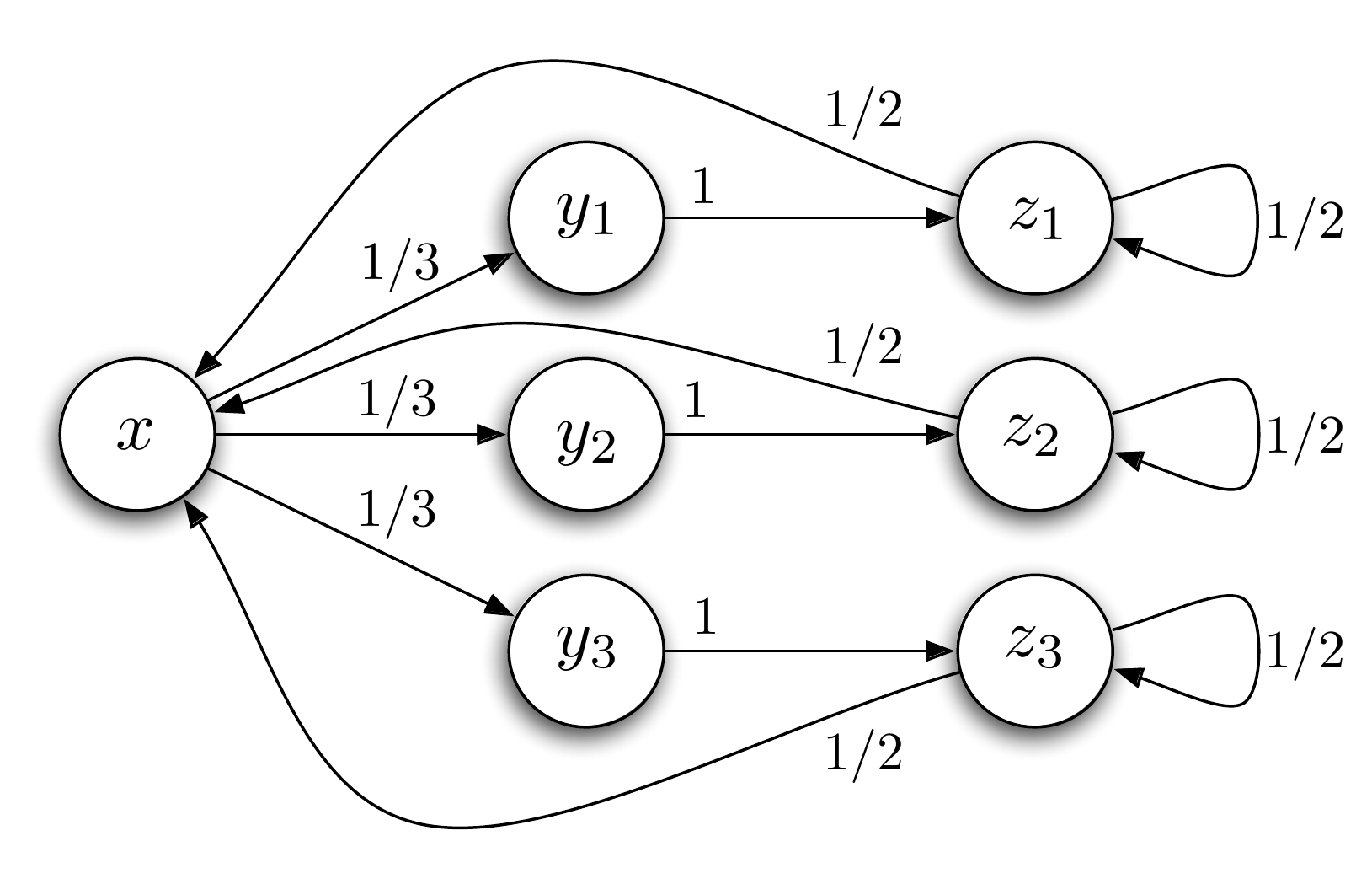}}
\end{minipage}}
\subfigure[$\WAbs_1$; $\sim_1\in PS \cap PW\cap CS\cap CW$]%
{\begin{minipage}{0.49\linewidth}
\label{pic2}
\mbox{\includegraphics[scale=0.2]{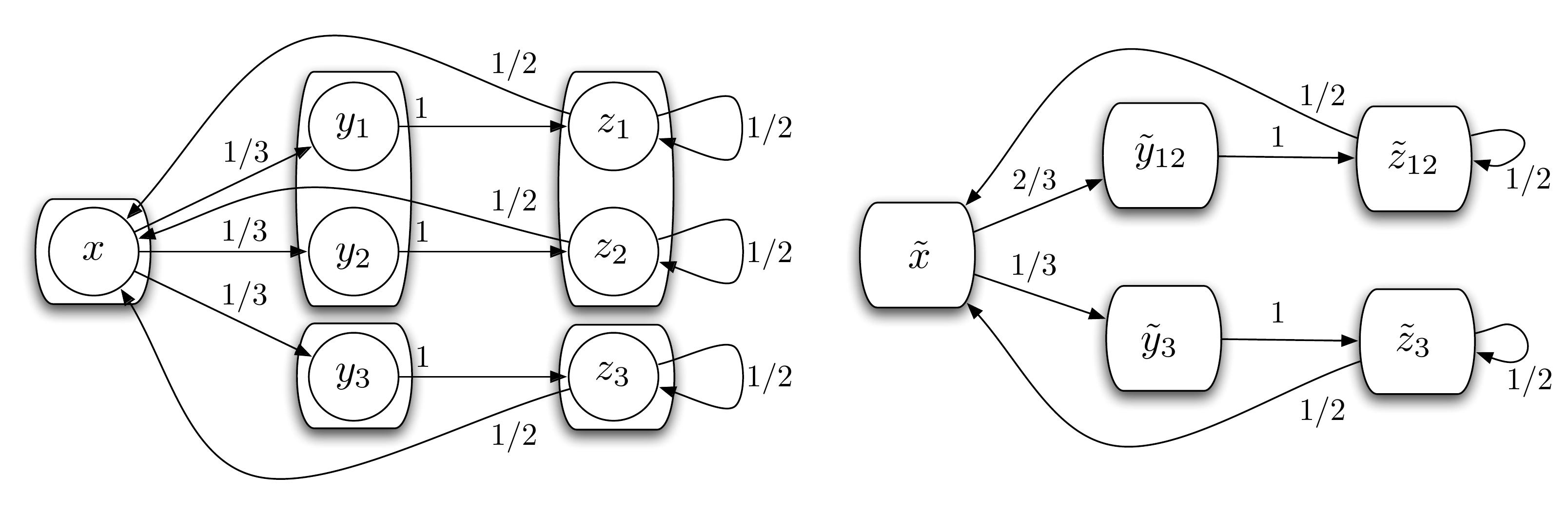}}
\end{minipage}}
\subfigure[$\WAbs_2$; $\sim_2\in CW\setminus CS$]%
{\begin{minipage}{0.49\linewidth}
\label{pic3}
\mbox{\includegraphics[scale=0.2]{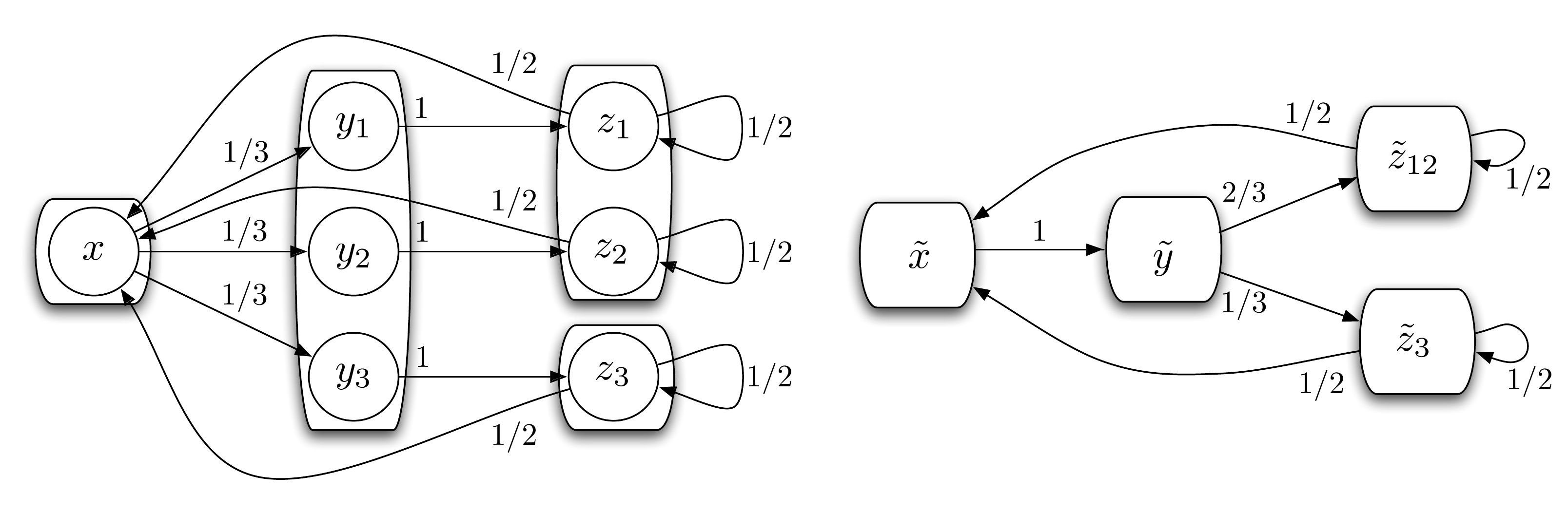}}
\end{minipage}}
\subfigure[$\WAbs_3$; $\sim_3\in CS\setminus CW$]%
{\begin{minipage}{0.49\linewidth}
\label{pic4}
\mbox{\includegraphics[scale=0.2]{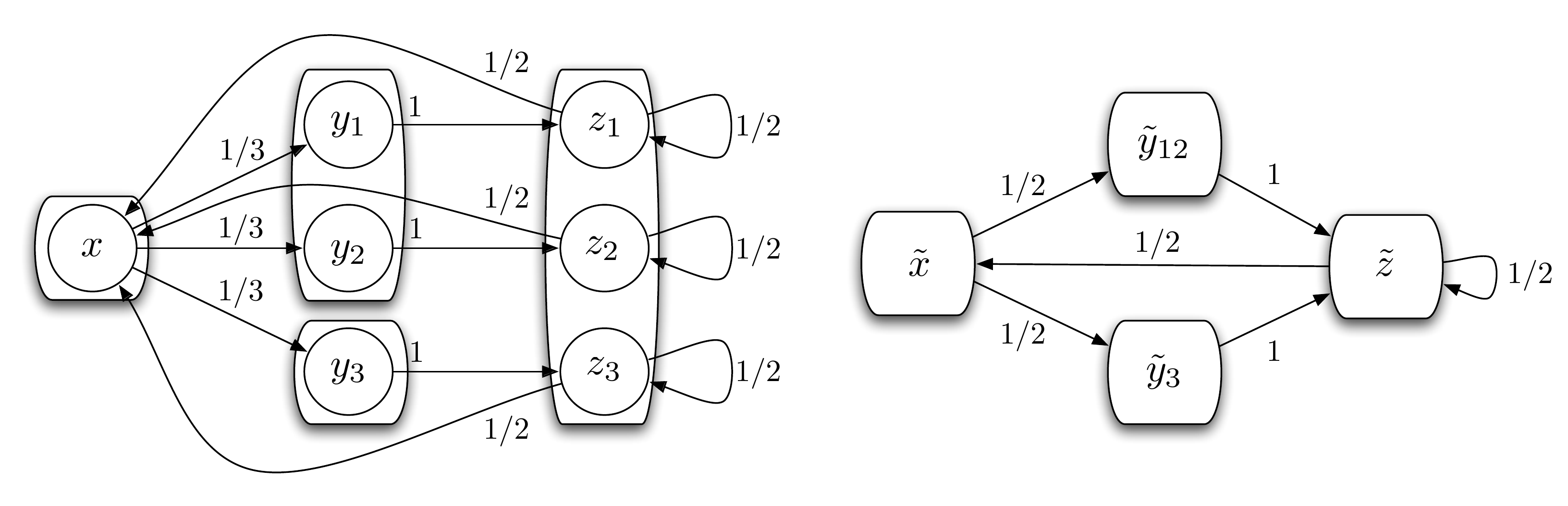}}
\end{minipage}}
\caption{Different abstractions of system ${\cal W}$}
\label{pics}
\end{figure}

This discussion indicates that if we decide to check for weak lumpability instead of for strong by using the characterization from Thm.~\ref{thm2},  it might happen that we eliminate the aggregations that are strongly lumpable. In the case of reductions of Kappa systems, we will use the weak lumpability characterization.

\subsection{Bisimulations}

Aiming to define the algorithm that is abstracting the WLTS of a Kappa system, we start by redefining the lumpability properties in the bisimulation notions. Bisimulation is typically defined by logically characterizing the distinguishing property of the states that may be aggregated.

We define three kinds of bisimulation relations on the WLTS, which are based on the lumpability characterizations given in Thm.~\ref{thm2}. We adopt the terminology of \cite{buchholz_bisimulation}. The \emph{forward} bisimulations arise from the characterization for strong lumpability: the bisimilar states have the same forward behaviour in the sense that they are each targeting any other lumped state with the same total affinity (total outgoing rate). This concept is well established  for dependability or performance analysis \cite{janePhD,holger_PhD}.  What we use in the abstractions of Kappa systems is \emph{backward} bisimulation. The bisimilar states have the same backward behaviour in the sense that they are reached by the predecessors from one lumped state with the same probabilistic quantity, which becomes the rate in the abstract system. It is however less established and only applied in very few approaches for stochastic modelling \cite{susanna}. The \emph{backward uniform} bisimulation is an instance of a backward bisimulation with an additional constraint that only the equally-probable states may be aggregated. 

\begin{definition} 
Given a WLTS ${\cal W} = (\S,\L,w,\pi_0)$, and 
$(\sim,\sim_L)$ a pair of equivalence relations respectively over $\S$ and $\L$, we define a function $\delta_F: {\S}\times {\L}_{/_{\sim_L}}\times \S_{/_{\sim}} \ra \RR_0^+$:
\[\delta_F(x_i,\tilde{l},\tilde{x}_j)  = \sum \{|w(x_i,l,x_j)|\;\mid\; l\in \tilde{l} \hbox{ and } x_j\in \tilde{x}_j\}.\]

Furthermore, given a family of probability distributions over the partitions $\ga\in \Gamma_{\S,{\sim}}$, we define the quantity $\delta_B: {\S}_{/_{\sim}}\times {\L}_{/_{\sim_L}} \times \S \ra \RR_0^+$:
\[\delta_B(\tilde{x}_i,\tilde{l},x_j) = \strutfrac{\sum \{\ga(\xAbs_i,x_i)\cdot |w(x_i,l,x_j)|\;\mid\; l\in \tilde{l},  x_i\in \tilde{x}_i\}}{\ga(\xAbs_j,x_j)}.\]

Specifically, if we have that $\ga$ is a uniform distribution over the equivalence classes, we can express the latter expression in terms of cardinalities of the equivalence classes:
\[\delta_{BU}(\tilde{x}_i,\tilde{l},x_j) = \frac{|\tilde{x}_j|}{|\tilde{x}_i|}\sum \{|w(x_i,l,x_j)|\;\mid\;l\in \tilde{l},  x_i\in \tilde{x}_i\}.\]
\end{definition}

\begin{definition} (Forward and  backward Markov bisimulation) 
Given a WLTS ${\cal W} = (\S,\L,w,\pi_0)$,  and 
$(\sim,\sim_L)$ a pair of equivalence relations respectively over $\S$ and $\L$, we say that $(\sim, \sim_L)$ is a
\begin{myenumerate}
\item 
\textit{Forward Markov Bisimulation}, if for all $x_i$ and $x_j$, the following is satisfied: $x_i \sim x_j$, iff for all equivalence classes $\tilde{x}\in \S_{/_{\sim}}$,$\tilde{l}\in\L_{/_{\sim_L}}$, we have that $a(x_i)=a(x_j)$  and $\delta_F(x_i,\tilde{l},\tilde{x}) = \delta_F(x_j,\tilde{l},\tilde{x})$.

\textit{Remark.} Note that this involves the bisimulation in the classical sense: if $x_i$ has a successor in some class, $x_j$ has it as well, and they are related by appropriate labels (and probabilities in this case).

\item 
\textit{Backward Markov bisimulation}, if for all $x_i$ and $x_j$, there exists an  $\ga\in\Gamma_{\S,\sim}$, such that the following is satisfied: $x_i \sim x_j$,  iff for all equivalence classes $\tilde{x}\in \S_{/_{\sim}}$, $\tilde{l}\in \L_{/_{\sim_L}}$,  we have that $a(x_i)=a(x_j)$ and $\delta_B(\tilde{x},\tilde{l},x_i) = \delta_B(\tilde{x},\tilde{l},x_j)$.

\end{myenumerate}
\end{definition}

\begin{thm} (Forward Markov bisimulation implies sound abstraction) 
Let ${\cal W} = (\S,\L,w,\pi_0)$ be a WLTS. 
If $(\sim,\sim_L)$ induces a forward Markov bisimulation, then for any
aggregates $\tilde{x}_i$, $\tilde{l}$, and $\tilde{x}_j$, we can define
\[
\tilde{w}(\tilde{x}_i,\tilde{l},\tilde{x}_j) = \delta_F(x_i,\tilde{l},\tilde{x}_j).
\]
The so defined abstraction $\tilde{{\cal W}} = (\S_{/_{\sim}},\L_{/_{\sim_L}},\tilde{w},\tilde{\pi_0})$ is sound.
We then say that ${\cal W}$ \emph{refines} $\tilde{\cal W}$ by a forward Markov bisimulation  $(\sim,\sim_L)$, written 
$
{\cal W}\preceq_{F,(\sim,\sim_L)} \tilde{\cal W}.
$
\end{thm}

\begin{thm} (Backward Markov bisimulation implies sound and complete abstraction) 
Given a WLTS ${\cal W} = (\S,\L,w,\pi_0)$, if $(\sim,\sim_L)$ induces a backward Markov bisimulation with conditional probabilities over the aggregates $\ga\in \Gamma_{\S,{\sim}}$, then for any aggregates $\tilde{x}_i$,$\tilde{l}$, and $\tilde{x}_j$, we can define
\begin{align}
\tilde{w}(\tilde{x}_i,\tilde{l},\tilde{x}_j) = \delta_B(\tilde{x}_i,\tilde{l},x_j).
\label{eqn1}
\end{align}
If $\ga(\tilde x) = \pi_0|_{\tilde{x}}$, then the so defined abstraction $\tilde{{\cal W}} = (\S_{/_{\sim}},\L_{/_{\sim_L}},\tilde{w},\tilde{\pi_0})$ is sound and complete. We then say that ${\cal W}$ refines $\tilde{\cal W}$ by a backward Markov bisimulation $(\sim,\sim_L)$ with conditional distributions $\ga$, written 
${\cal W}\preceq_{B,(\sim,\sim_L),\ga} \tilde{\cal W}$.

In particular, if we know that $\ga$ is uniform, the equation (\ref{eqn1}) becomes $\tilde{w}(\tilde{x}_i,\tilde{l},\tilde{x}_j) = \delta_{BU}(x_i,\tilde{l},\tilde{x}_j)$, written also  ${\cal W}\preceq_{BU,(\sim,\sim_L)} \tilde{\cal W}$.
\end{thm}

\subsection{Proving bisimulations}

The forward bisimulation relation for abstracting the transition systems with CTMC semantics has been established and used in applications (eg, \cite{janePhD, holger_PhD}).  Moreover, computing the backward uniform bisimulation when $\ga$ is uniform is defined \cite{buchholz_bisimulation, susanna}.  It is based on an alternative characterization of the backward uniform Markov bisimulation, which eases the analysis.

\begin{lemma} (Proving backward uniform Markov bisimulation) \label{lem:main}
Let ${\cal W} = (\S,\L,w,\pi_0)$ be a WLTS and  $(\sim,\sim_L)$ be a pair of equivalence relations respectively over $\S$ and $\L$. For any state $x'\in\S$, and any pair of classes $\xAbs,\lAbs \in\ \S_{/\sim}\times \L_{/\sim_L}$, let us define the set $\mathrm{Pred}(\xAbs,\lAbs,x')$ of transitions from a state in $\xAbs$ to the state $x'$ and with a label in $\lAbs$ as follows:
\[\mathrm{Pred}(\xAbs,\lAbs,x') = \{(x,l)\in\xAbs\times\lAbs \;\mid\; w(x,l,x')>0\}.\]
Assume that: (1) $\pi_0|_{\S_{/\sim}}= \tilde{\pi_0}$, and (2) for any $x'_i,x'_j\in \S$ such that $x'_i \sim x'_j$ and any $\xAbs\in \S_{/\sim}$, $\lAbs\in \L_{/\sim_L}$, there exists a bijective map $\phi$ between $\mathrm{Pred}(\xAbs,\lAbs,x'_i)$ and $\mathrm{Pred}(\xAbs,\lAbs,x'_j)$,  such that for any $(x_i,l_i)\in \mathrm{Pred}(\xAbs,\lAbs,x'_i)$, if $\phi(x_i,l_i)=(x_j,l_j)$, then we have that $w(x_i,l_i,x'_i)=w(x_j,l_j,x'_j)$. 

Then we have that  ${\cal W}$ is the backward uniform bisimulation of the abstraction $\tilde{\cal W} = (\S_{/\sim},\L_{/\sim},\tilde{w},\tilde{\pi_0})$, i.e.~${\cal W}\preceq_{BU,(\sim,\sim_L)} \tilde{\cal W}$.
\end{lemma}

On the other hand, as soon as $\gamma$ over the aggregates is not uniform, we cannot observe the bijection between  predecessors over the states. Proving that the given abstraction is a backward bisimulation cannot be established unless we have a right 'guess' of the distributions $\ga$. Lem.~\ref{lem4} states how to avoid proving backward bisimulation by instead proving two uniform backward bisimulations. More precisely, if we want to prove the backward refinement between the systems  ${\cal W}$ and ${\WAbs}$, it is enough to observe the system  ${\cal W}^i$, which is a backward uniform refinement of both of the systems ${\cal W}$ and ${\WAbs}$ (Fig.~\ref{fig:lem3}).

\begin{lemma} (Proving backward Markov bisimulation) \label{lem4}
Consider a WLTS ${\cal W}=(\S,\L,w,\pi_{0})$, and any aggregation relation $(\sim,\sim_L)$, such that $\tilde{\cal W}=(\S_{/\sim},\L_{/\sim_L},\tilde{w},\tilde{\pi}_0)$. We assume that there exist a system ${\cal W}^i = (\S^i,\L^i,w^i,\pi_0^i)$, and the two pairs of equivalence relations $(\sim_1,\sim_{L1})$, $(\sim_2,\sim_{L2})$, such that ${\cal W}^i\preceq_{BU,(\sim_1,\sim_{L1})} {\cal W}$, ${\cal W}^i\preceq_{BU,(\sim_2,\sim_{L2})} \WAbs$, $\sim_1\preceq \sim_2$ (in the sense that, for any $x^i_1,x^i_2\in \S^i$,  $x^i_{1}\sim_1 x^i_{2} \Rar x^i_{1}\sim_2 x^i_{2}$), $\sim_{L1} \preceq \sim_{L2}$, and, for any $x^i_1,x^i_2$ such that $x^i_1 \sim_2 x^i_2$, the number of states which are $\sim_{L1}$- equivalent to $x^i_1$ is equal to the number of states which are $\sim_{L1}$- equivalent to $x^i_2$. Under this assumption, we have that ${\cal W}\preceq_{B,(\sim,\sim_L),\ga} \WAbs$, where $\ga$ is defined as 
\[
\ga(\tilde{x},x) = \frac{|\{x^i\in \S^i\;\mid\; x^i\sim_1 x_0^i\}|}  {|\{x^i\in {\S}^i\;\mid\; x^i\sim_2 {x}_0^i\}|}, \hbox{ for any } [x_0^i]_{\sim 1} =x.
\]
\end{lemma}

This Lemma contains the key observation  for the abstraction of Kappa systems, and for proving Thm.\ref{thm:main}. It thus completes the intention of the theoretical analysis in this paper. More precisely, we observe the WLTS ${\cal W}^{\cal R}$ of a given Kappa system ${\cal R}$, as defined in Dfn.~\ref{dfn:WLTSKappa} and its abstraction generated as proposed in the reduction procedure (Sec.~\ref{sec:reduction}, Dfn.~\ref{dfn:WLTSKappaAbs}). The main observation is that the system ${\cal W}$ is already an abstraction. More concretely, the states of ${\cal W}$ are multisets of species, and as such, they abstract the \emph{individual} species. For example, a state that contains two agents of type $A(s_u)$ abstracts away the potential individual behavior of these two agents, for example $A_1(s_u)$ and $A_2(s_u)$. To show that the abstraction is sound and complete, we observe the system ${\cal W}^i$, which is the \emph{individual-based} semantics of a Kappa system, where each individual agent is uniquely identified. The backward uniform refinement is established between ${\cal W}^i$ and ${\cal W}$ by the modeling assumptions. We are left to prove the backward uniform refinement between ${\cal W}^i$ and $\WAbs$. This is done by inspections on the ACM's (Dfn.~\ref{dfn:ACM}).

\begin{figure}\label{fig:lem3}
\hspace*{4cm} \includegraphics[scale=0.3]{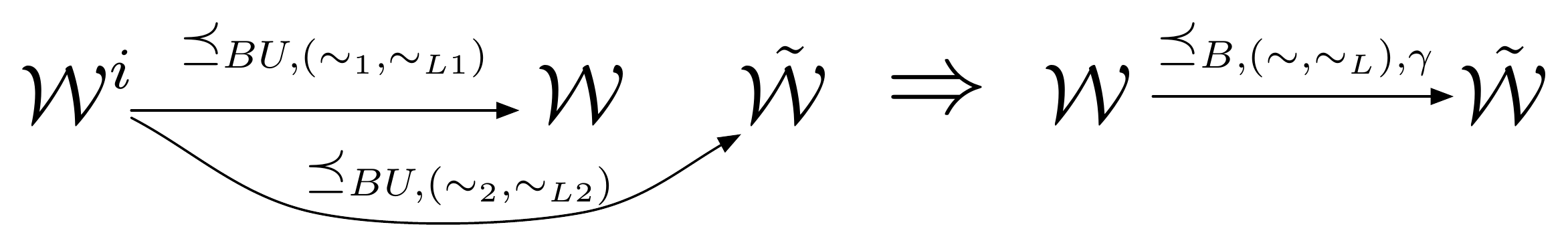} 
\caption{Proving backward refinement}
\end{figure}

\subsection{Example}

We consider the following  Kappa system. We have the agent types ${\cal A}=\{A,B\}$, the site names $\{s,t\}$, the signatures $\Sigma_\iota(A) = \Sigma_\iota(B) = \{s\}$ and $\Sigma_\lambda(A)= \Sigma_\lambda(B) =  \{t\}$,  the alphabet of internal states ${\mathbb I} =\{u,p\}$. The contact map is defined by $(\cal N,\cal E)$, such that ${\cal N} = \{(A,s),(A,t),(B,s),(B,t)\}$ and ${\cal E} = \{((A,t),(B,t))\}$ and the following rules:
\begin{align*}
r1 :\; & \agent{A}{\site{s}{u}{}} \lrar \agent{A}{\site{s}{p}{}} @ k_1,k_{1-} \\
r2 :\;& \agent{B}{\site{s}{u}{}} \lrar \agent{B}{\site{s}{p}{}} @ k_2,k_{2-} \\
r3 :\;& \agent{A}{\site{t}{}{}}\agentsep\agent{B}{\site{t}{}{}} \lrar 
\agent{A}{\site{t}{}{1}}\agentsep\agent{B}{\site{t}{}{1}} @ k_3,k_{3-} 
\end{align*}
Moreover, using the minimal ACM for annotating the agents, as written in Dfn.~\ref{dfn:ACM}, we get that $\approx_A$ has two equivalence classes $\{s\}$ and $\{t\}$; and that $\approx_B$ has two equivalence classes $\{s\}$ and $\{t\}$ as well. 

The fragments derived from an ACM (Dfn.~\ref{dfn14}) are the following: $F_1= \agent{A}{\site{s}{u}{}}$, $F_2= \agent{A}{\site{s}{p}{}}$, $F_3= \agent{A}{\site{t}{}{}}$,  $F_4= \agent{A}{\site{t}{}{1}}\agentsep \agent{B}{\site{t}{}{1}}$, $F_5= \agent{B}{\site{s}{u}{}}$,  $F_6= \agent{B}{\site{s}{p}{}}$, $F_7= \agent{B}{\site{t}{}{}}$. 

Let us pick a (finite) initial distribution $\pi_0$. Now we observe the WLTS ${\cal W} = (\S,\L,w,\pi_0)$ assigned to the Kappa system ${\cal R}_{AB}$ (introduced in Dfn.~\ref{dfn9}), and the state $y$ which is the $\equiv$-equivalence class of the mixture $E_y$ defined as follows:
\[\text{\agent{A}{\!\site{s}{p}{},\site{t}{}{1}}\!\agentsep\!\agent{B}{\!\site{s}{p}{},\site{t}{}{1}}\!\agentsep\!\agent{A}{\!\site{s}{u}{},\site{t}{}{2}}\!\agentsep\!\agent{B}{\!\site{s}{u}{},\site{t}{}{2}}\!\agentsep\!\agent{A}{\!\site{s}{u}{},\site{t}{}{3}}\!\agentsep\!\agent{B}{\!\site{s}{u}{},\site{t}{}{3}}}\]
The unique (up to $\equiv)$ non $\equiv^\sharp$-equivalent mixtures is $E_{y'}$ which is defined as follows:
\[\text{\agent{A}{\!\site{s}{p}{},\site{t}{}{1}}\!\agentsep\!\agent{B}{\!\site{s}{u}{},\site{t}{}{1}}\!\agentsep\!\agent{A}{\!\site{s}{u}{},\site{t}{}{2}}\!\agentsep\!\agent{B}{\!\site{s}{p}{},\site{t}{}{2}}\!\agentsep\!\agent{A}{\!\site{s}{u}{},\site{t}{}{3}}\!\agentsep\!\agent{B}{\!\site{s}{u}{},\site{t}{}{3}}}\]
We denote $y'=[E_{y'}]_{\equiv}$, $\tilde{y'}=[E_{y'}]_{\equiv^\sharp}$. We compute however that the distribution among state $\tilde{y}=[E_y]_{\equiv^\sharp}$. 

We have: $\ga(\tilde{y},y)=1/3$ and $\ga(\tilde{y},y')=2/3$. 
Roughly speaking this comes from the fact that if we annotate fragments of type $A$ and $B$ in $\tilde{y}$ with the identifiers $1$, $2$, $3$ (there are $36$ possible annotation), and if we assume that agents with the same identifiers are bound together. Then there are $12$ annotations so that the phosphorilated $A$ and $B$ are bound together, and $24$ where this is not the case. 
A more detailed analysis of this model is given in \cite{journal}.

\section{Case study}\label{sec:crosstalk}

In this section, we apply our framework to a case study. We have indeed refactored in Kappa the model of a crosstalk between the \agentEGF{} rececptor and the Insulin receptor pathways which is described in \cite{conzelmann2008}. In this model, two kinds of receptors, \agentEGF{} receptor (\agentEGFR{}) and insulin receptor (\agentIR{}) can recruit a protein called \agentSos{}, which can be phosphorylated, or not. Each kind of receptor has its own pathway, and these two pathways shared some common proteins.

The contact map is given in Fig.~\ref{CM}. One can notice that some sites can be bound to several other sites, which denotes a competition (concurrency). Moreover, the site $d$ of a \agentEGF{} receptor (\agentEGFR{}) can be bound to the site $d$ of another \agentEGFR{} (since there is a loop in the CM). Moreover rules are given in Table \ref{ruleset}. We do not give the rate for rules, but we assume no hypotheses on the rates (rates are parameters, some of them might be equal, or not). Moreover, each rule is reversible.

\begin{figure}[t]
{\begin{minipage}{\linewidth}
\begin{center}
\scalebox{0.5}{\begin{minipage}{\linewidth}
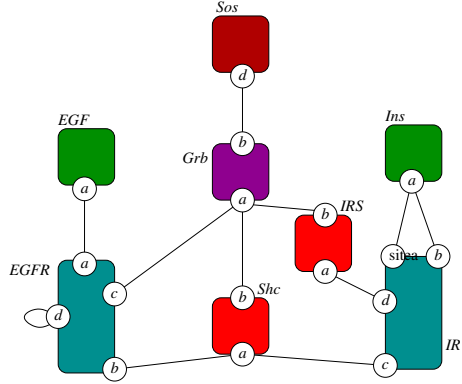
\end{minipage}}
\end{center}
\end{minipage}}
\caption{Contact map for the \emph{EGFR}/Insulin crosstalk.}
\label{CM}
\end{figure}

{\begin{table}[thp]
\begin{minipage}{0.49\linewidth}
\textmode                                    
\putrulelabels
\changeruleprefix{r} 
\scalebox{0.73}
{\begin{minipage}{1.4\linewidth}
\firstsep
\agentEGF{\sitea{}{}}\agentsep\agentEGFR{\sitea{}{}\sitesep\sited{}{}}\revrule\agentEGF{\sitea{}{1}}\agentsep\agentEGFR{\sitea{}{1}\sitesep\sited{}{}}\sep
\agentEGF{\sitea{}{}}\agentsep\agentEGFR{\sitea{}{}\sitesep\sited{}{\bound{1}}}\revrule\agentEGF{\sitea{}{1}}\agentsep\agentEGFR{\sitea{}{1}\sitesep\sited{}{\bound{2}}}\sep
\agentEGFR{\sitea{}{}\sitesep\sited{}{}}\agentsep\agentEGFR{\sitea{}{\bound{1}}\sitesep\sited{}{}}\revrule\agentEGFR{\sitea{}{}\sitesep\sited{}{1}}\agentsep\agentEGFR{\sitea{}{\bound{2}}\sitesep\sited{}{1}}\sep
\agentEGFR{\sitea{}{}\sitesep\sited{}{}}\agentsep\agentEGFR{\sitea{}{}\sitesep\sited{}{}}\revrule\agentEGFR{\sitea{}{}\sitesep\sited{}{1}}\agentsep\agentEGFR{\sitea{}{}\sitesep\sited{}{1}}\sep
\agentEGFR{\sitea{}{\bound{1}}\sitesep\sited{}{}}\agentsep\agentEGFR{\sitea{}{\bound{2}}\sitesep\sited{}{}}\revrule\agentEGFR{\sitea{}{\bound{2}}\sitesep\sited{}{1}}\agentsep\agentEGFR{\sitea{}{\bound{3}}\sitesep\sited{}{1}}\sep
\agentEGFR{\siteb{u}{}\sitesep\sited{}{}}\revrule\agentEGFR{\siteb{p}{}\sitesep\sited{}{}}\sep 
\agentEGFR{\siteb{u}{}\sitesep\sited{}{\bound{1}}}\revrule\agentEGFR{\siteb{p}{}\sitesep\sited{}{\bound{1}}}\sep
\agentEGFR{\siteb{p}{}}\agentsep\agentShc{\sitea{}{}}\revrule\agentEGFR{\siteb{p}{1}}\agentsep\agentShc{\sitea{}{1}}\sep
\agentEGFR{\siteb{p}{1}\sitesep\sited{}{}}\agentsep\agentShc{\sitea{}{1}\sitesep\siteb{u}{}}\revrule\agentEGFR{\siteb{p}{1}\sitesep\sited{}{}}\agentsep\agentShc{\sitea{}{1}\sitesep\siteb{p}{}}\sep
\agentEGFR{\siteb{p}{1}\sitesep\sited{}{\bound{2}}}\agentsep\agentShc{\sitea{}{1}\sitesep\siteb{u}{}}\revrule\agentEGFR{\siteb{p}{1}\sitesep\sited{}{\bound{2}}}\agentsep\agentShc{\sitea{}{1}\sitesep\siteb{p}{}}\sep
\agentGrb{\sitea{}{}}\agentsep\agentShc{\siteb{p}{}}\revrule\agentGrb{\sitea{}{1}}\agentsep\agentShc{\siteb{p}{1}}\sep
\agentEGFR{\sitec{u}{}\sitesep\sited{}{}}\revrule\agentEGFR{\sitec{p}{}\sitesep\sited{}{}}\sep
\agentEGFR{\sitec{u}{}\sitesep\sited{}{\bound{1}}}\revrule\agentEGFR{\sitec{p}{}\sitesep\sited{}{\bound{1}}}\sep
\agentEGFR{\sitec{p}{}\sitesep\sited{}{}}\agentsep\agentGrb{\sitea{}{}}\revrule\agentEGFR{\sitec{p}{1}\sitesep\sited{}{}}\agentsep\agentGrb{\sitea{}{1}}\sep
\agentEGFR{\sitec{p}{}\sitesep\sited{}{\bound{1}}}\agentsep\agentGrb{\sitea{}{}}\revrule\agentEGFR{\sitec{p}{1}\sitesep\sited{}{\bound{2}}}\agentsep\agentGrb{\sitea{}{1}}\sep
\agentIR{\sitea{}{}\sitesep\siteb{}{}}\agentsep\agentIns{\sitea{}{}}\revrule\agentIR{\sitea{}{1}\sitesep\siteb{}{}}\agentsep\agentIns{\sitea{}{1}}\sep
\agentIR{\sitea{}{}\sitesep\siteb{}{\bound{1}}}\agentsep\agentIns{\sitea{}{}}\revrule\agentIR{\sitea{}{1}\sitesep\siteb{}{\bound{2}}}\agentsep\agentIns{\sitea{}{1}}\sep
\agentIR{\sitea{}{}\sitesep\siteb{}{}}\agentsep\agentIns{\sitea{}{}}\revrule\agentIR{\sitea{}{}\sitesep\siteb{}{1}}\agentsep\agentIns{\sitea{}{1}}\sep
\agentIR{\sitea{}{\bound{1}}\sitesep\siteb{}{}}\agentsep\agentIns{\sitea{}{}}\revrule\agentIR{\sitea{}{\bound{2}}\sitesep\siteb{}{1}}\agentsep\agentIns{\sitea{}{1}}

\end{minipage}}
\end{minipage}
\begin{minipage}{0.49\linewidth}
\scalebox{0.73}
{\begin{minipage}{1.4\linewidth}
\textmode                                    
\putrulelabels
\changeruleprefix{r} 
\sep
\agentIR{\sitea{}{}\sitesep\siteb{}{}\sitesep\sitec{u}{}}\revrule\agentIR{\sitea{}{}\sitesep\siteb{}{}\sitesep\sitec{p}{}}\sep
\agentIR{\sitea{}{\bound{1}}\sitesep\siteb{}{}\sitesep\sitec{u}{}}\revrule\agentIR{\sitea{}{\bound{1}}\sitesep\siteb{}{}\sitesep\sitec{p}{}}\sep
\agentIR{\sitea{}{}\sitesep\siteb{}{\bound{1}}\sitesep\sitec{u}{}}\revrule\agentIR{\sitea{}{}\sitesep\siteb{}{\bound{1}}\sitesep\sitec{p}{}}\sep
\agentIR{\sitea{}{\bound{1}}\sitesep\siteb{}{\bound{2}}\sitesep\sitec{u}{}}\revrule\agentIR{\sitea{}{\bound{1}}\sitesep\siteb{}{\bound{2}}\sitesep\sitec{p}{}}\sep
\agentIR{\sitec{p}{}}\agentsep\agentShc{\sitea{}{}}\revrule\agentIR{\sitec{p}{1}}\agentsep\agentShc{\sitea{}{1}}\sep
\agentIR{\sitea{}{\bound{2}}\sitesep\siteb{}{\bound{3}}\sitesep\sitec{}{1}}\agentsep\agentShc{\sitea{}{1}\sitesep\siteb{u}{}}\revrule\agentIR{\sitea{}{\bound{2}}\sitesep\siteb{}{\bound{3}}\sitesep\sitec{}{1}}\agentsep\agentShc{\sitea{}{1}\sitesep\siteb{p}{}}\sep
\agentIR{\sitea{}{}\sitesep\siteb{}{}\sitesep\sited{u}{}}\revrule\agentIR{\sitea{}{}\sitesep\siteb{}{}\sitesep\sited{p}{}}\sep
\agentIR{\sitea{}{\bound{1}}\sitesep\siteb{}{}\sitesep\sited{u}{}}\revrule\agentIR{\sitea{}{\bound{1}}\sitesep\siteb{}{}\sitesep\sited{p}{}}\sep
\agentIR{\sitea{}{}\sitesep\siteb{}{\bound{1}}\sitesep\sited{u}{}}\revrule\agentIR{\sitea{}{}\sitesep\siteb{}{\bound{1}}\sitesep\sited{p}{}}\sep
\agentIR{\sitea{}{\bound{1}}\sitesep\siteb{}{\bound{2}}\sitesep\sited{u}{}}\revrule\agentIR{\sitea{}{\bound{1}}\sitesep\siteb{}{\bound{2}}\sitesep\sited{p}{}}\sep
\agentIR{\sited{p}{}}\agentsep\agentIRS{\sitea{}{}}\revrule\agentIR{\sited{p}{1}}\agentsep\agentIRS{\sitea{}{1}}\sep
\agentIR{\sitea{}{\bound{2}}\sitesep\siteb{}{\bound{3}}\sitesep\sited{}{1}}\agentsep\agentIRS{\sitea{}{1}\sitesep\siteb{u}{}}\revrule\agentIR{\sitea{}{\bound{2}}\sitesep\siteb{}{\bound{3}}\sitesep\sited{}{1}}\agentsep\agentIRS{\sitea{}{1}\sitesep\siteb{p}{}}\sep
\agentGrb{\sitea{}{}}\agentsep\agentIRS{\siteb{p}{}}\revrule\agentGrb{\sitea{}{1}}\agentsep\agentIRS{\siteb{p}{1}}\sep
\agentGrb{\siteb{}{}}\agentsep\agentSos{\sited{u}{}}\revrule\agentGrb{\siteb{}{1}}\agentsep\agentSos{\sited{u}{1}}\sep
\agentGrb{\siteb{}{}}\agentsep\agentSos{\sited{p}{}}\revrule\agentGrb{\siteb{}{1}}\agentsep\agentSos{\sited{p}{1}}\sep
\agentGrb{\siteb{}{1}}\agentsep\agentSos{\sited{u}{1}}\revrule\agentGrb{\siteb{}{1}}\agentsep\agentSos{\sited{p}{1}}\sep
\agentSos{\sited{u}{}}\revrule\agentSos{\sited{p}{}}\sep
\agentShc{\siteb{u}{}}\revrule\agentShc{\siteb{p}{}}\sep
\agentIRS{\siteb{u}{}}\revrule\agentIRS{\siteb{p}{}}\finalsep

\end{minipage}}
\end{minipage}
\caption{Rule set for the \emph{EGFR}/Insulin crosstalk.}
\label{ruleset}
\end{table}}

We roughly explain how each pathway works, by focusing on the forward direction of rules. First, we describe how \agentEGFR{} can recruit a transport molecule (\agentGrb{}). \agentEGFR{} can recruit a ligand (EGF) on site $a$ (r01,r02), and two \agentEGFR{}s can form a dimer (r03,r04,r05). We have used two rules to encode EGF/\agentEGFR{} binding, in order to model the fact that the rate of association may depend on the fact whether \agentEGFR{} is in a dimer, or not. The same way, the rate of dimer formation/dissociation may depend on the number of ligands that are bound to receptors. The site $b$ of \agentEGFR{} can be phosphorylated (r06,r07) at a rate which depends whether the receptor is in a dimer, or not. Then, \agentEGFR{} can recruit an adapter molecule called Shc (r08). Then, \agentEGFR{} can phosphorylize Shc (r09,r10) (the rate depend on the fact whether the receptor is still in a dimer, or not). Shc can then recruit a transport modecule (\agentGrb{}) (r11). Yet, each receptor  has a shorter way to recruit a transport molecule. The site $c$ of \agentEGFR{} can be phosphorylated (r12,r13), and then recruit \agentGrb{} directly (r14,r15). 

Then, we describe how an insulin receptor (\agentIR{}) can recruit \agentGrb{}. \agentIR{} can recruit insulin molecules (Ins) on two sites $a$ (r16,r17) and $b$ (r18,r19) (the rate may depend on the fact whether an insulin molecule has already been recruited).  The site $c$ of the \agentIR{} can be phosphorylated (r20,r21,r22,r23) at a rate which depends on the number of recruited insulin molecules (in practice the rates of rules r21 and r22 are the same). Then, \agentIR{} can recruit an adapter Shc (r24). Whenever \agentIR{} is also bound to two insulin molecules, Shc can be phosphorylized (r25). Shc can then recruit \agentGrb{} (r11). Yet, \agentIR{} has an other way of recruiting a \agentGrb{}. The site $d$ of \agentIR{} (r26,r27,r28,r29), and then recruit another adapter called \agentIRS{} (r30) which can be activated when the insulin receptor is bound to two insulin molecules (r31). Then, \agentIRS{} can recruit \agentGrb{} (r32).  

Independently, \agentGrb{} can recruit a protein \agentSos{} (r33,r34). And \agentSos{} can be activated (r35,r36) at the rate which may depend on the fact whether it is bound to a \agentGrb{}, or not. Other rules describe the recuitment of \agentSos{} by \agentGrb{}. And spontaneous (de)phosphorylation of Shc (r37) and \agentIR{}S (r38).

In this model, $2,768$ different complexes may occur. This number is mainly due to the fact that each dimer made of two proteins \agentEGFR{} has $4$ sites (the sites \site{b}{}{} and \site{c}{}{} for each \agentEGFR{}) to recruit a \agentGrb{}, which induces a small combinatorial blow up. Scanning the set of rules, one can notice that no rule tests both the site $a$ and $b$ of some  proteins \agentGrb{}. Thus, the partition $\{\{a\},\{b\}\}$ can be used safely for the sites of \agentGrb{}, in the annotated contact map. As a consequence, the number of fragments is only $609$. Unfortunatly, this is the only reduction that we can do (ie, the partition for the sites of any other kind of proteins, is the coarsest one).

Last, one can notice that, given some additional hypotheses on the rate of some rule, that the sites $a$ and $b$ of \agentIR{} have a symetric role in the system. We could consider this symetry to reduce the set of considered fragments, by identifying two symetric fragments, such as \agent{IR}{\site{a}{}{1}\sitesep\site{b}{}{}}\agentsep\agent{Ins}{\site{a}{}{1}} and \agent{IR}{\site{a}{}{}\sitesep\site{b}{}{1}}\agentsep\agent{Ins}{\site{a}{}{1}}.

\section{Conclusions}\label{sec:conclusion}
Reducing the complexity of combinatorial reaction mixtures is an important milestone towards simulation and analysis of large-scale realistic models of cellular signal transduction. In this paper we study a scalable reduction method, that is applicable to any rule-based specification. The reduction is sound and moreover complete, i.e. the sample traces of individual molecular species can be reconstructed from the traces of aggregated species in the reduced model. We put this method into the general context of abstractions of probabilistic transition system and show that it yields a sufficient condition for weak lumpability and that it is equivalent to backward Markov bisimulation.  The reduction factor depends on the number of independent molecular events and is strictly smaller than that of the less-demanding reduction based on the differential semantics.  

A compelling problem for future work is thus to analyze differential fragments in the context of stochastic semantics and to obtain error bounds for this reduction as a function of the kinetic parameters of the system.

{
\bibliographystyle{eptcs} 
\bibliography{ref} 
}


\end{document}